\documentclass[lettersize,journal]{IEEEtran}
\usepackage{amsmath,amsfonts}
\usepackage{algorithmic}
\usepackage{algorithm}
\usepackage{array}
\usepackage[caption=false,font=normalsize,labelfont=sf,textfont=sf]{subfig}
\usepackage{textcomp}
\usepackage{stfloats}
\usepackage{url}
\usepackage{verbatim}
\usepackage{graphicx}
\usepackage{cite}
\usepackage{tabularx}
\usepackage{amsthm}
\newtheorem{lemma}{Lemma}

\usepackage{xcolor}

\begin{document}

\title{Degree Centrality Algorithms for Weighted Multilayer Networks (or $\omega$-MLNs)}

\author{Ayomide Ayowole-Obi, Abhishek Santra, Sharma Chakravarthy}

\markboth{Journal of \LaTeX\ Class Files,~Vol.~14, No.~8, August~2021}%
{Shell \MakeLowercase{\textit{et al.}}: A Sample Article Using IEEEtran.cls for IEEE Journals}


\maketitle

\begin{abstract}
Centrality measures are defined for simple graphs -- directed, undirected, weighted or unweighted.  Attributed graphs have 
to be reduced to simple graphs for computing centrality measures. However, when applications with multiple types of relationships are modeled using multilayer networks (MLNs), simple graph algorithms cannot be \textit{directly} used. Existing approaches typically analyze MLNs by aggregating layers of an MLN into a single graph, which results in the loss of structural and semantic information. The semantic information loss can be more pronounced particularly, in weighted networks. 

This work focuses on computing \textit{degree centrality} in weighted homogeneous multilayer networks (HoMLNs) using a decoupling-based framework. The framework performs independent layer-wise analysis on MLNs without reducing them to simple graphs. The decoupling approach allows  use of exiting algorithms for each layer and uses minimal information from individual layers for computing degree centrality of HoMLNs.  We propose heuristic-based algorithms that strike a balance between accuracy and efficiency.

The proposed methods are evaluated against ground truth (GT) results obtained using Boolean OR aggregation and naive baselines. Experimental results on both synthetic and real-world HoMLN datasets demonstrate that the heuristics achieve accuracy comparable to the ground truth while significantly improving computational efficiency, thereby establishing the scalability and effectiveness of the HoMLN algorithms developed using the decoupling approach.
\end{abstract}


\section{INTRODUCTION}
\label{section:introduction}

Data with multiple relationships between entities can be represented using different models, one of which is a graph. Graphs consist of nodes (vertices) representing entities and edges (links) to denote their relationship. Edges may carry weights that quantify the value or cost of a connection. Graphs are widely used because they support relationships and are easy to visualize. Also, many algorithms are available for their analysis. However, simple graphs struggle to capture multiple relationships between the same entities. Multilayer Networks (MLNs) can handle this by representing each relationship type using a separate layer of a MLN, thereby preserving semantics and providing a richer and comprehensive model for analysis~\cite{DBLP:journals/fdata/SantraIMCB23-short,DBLP:conf/iccS/SantraBC25-short,CommSurveyKimL15}. 

Graph analysis often utilizes metrics such as community detection, clustering, substructure identification, and centrality to name a few. Characteristics and attributes of MLN layers differ in HoMLNs in terms connections between nodes in different layers. Centrality metrics, such as degree, closeness, and betweenness, measure node importance -- both locally or globally -- within a network. Freeman~\cite{freeman1978centrality} formalized these measures after surveying and analyzing published papers on communication networks. Degree centrality (DC) specifically evaluates a node's importance based on its connections~\cite{6085951}. In weighted networks, this concept is extended as node strength, which is calculated by considering both the number of edges and their weights~\cite{opsahl2010node}. Unlike global metrics such as closeness and betweenness, DC is a local measure. 

A weighted multilayer network (MLN)\cite{KDIR2022/PavelSC22} consists of layers, each representing a simple graph with nodes representing entities and edges representing relationships. Within each layer, nodes are connected by intra-layer edges that may have  weights that capture the significance of the interaction. 

Below are formal definitions of weighted multilayer networks. 



\textbf{Definition 1.} 
A weighted multilayer network \textit{$\omega$-MLN}(G, X) is defined by two sets of graphs.  
The set $G = \{G_1, G_2, \dots, G_n\}$ contains \textit{edge-weighted simple graphs},  
where each layer simple graph $G_i = (V_i, E_i)$ is defined by a set of vertices $V_i$ and a set of edges $E_i$.  
An edge $(u, v) \in E_i$ connects vertices $u$ and $v$, where $u, v \in V_i$.  
Each graph $G_i$ is associated with a non-negative weight function:
\begin{equation*}
    W_i : E_i \rightarrow \mathbb{R}_{\ge 0}
\end{equation*}
    
where $W_i(u, v)$ denotes the weight of the edge $(u, v)$ in layer $i$.

The set of inter-layer graphs 
\begin{equation*}
    X = \{X_{1,2}, X_{1,3}, \dots, X_{n-1,n}\}
\end{equation*}
consists of \textit{bipartite graphs}.  
Each graph $X_{i,j} = (V_i, V_j, L_{i,j})$ is defined by vertex sets $V_i$ and $V_j$ and a set of links $L_{i,j}$,  
such that for every edge $l(a,b) \in L_{i,j}$, $a \in V_i$ and $b \in V_j$.

If inter-layer links are weighted, a corresponding weight function
\begin{equation*}
    W_{i,j} : L_{i,j} \rightarrow \mathbb{R}_{\ge 0}
\end{equation*}
may be defined to assign weights to the links in $X_{i,j}$.


In a weighted homogeneous multilayer network ($\omega$-HoMLN), the vertex sets across all layers have large overlap. Common set of vertices appears in each layer, but the edges and their corresponding weights may differ based on the type, frequency, or intensity of interactions represented in each layer. 

Traditional approaches for analyzing multilayer networks often convert an MLN into an aggregated simple graph using operations such as Boolean AND/OR on layer edges. Aggregation of weighted layers are a bit more complicated. In addition to Boolean AND/OR operations, a function for aggregating edge weights is needed. In the case of weighted multilayer networks, this might include operations such as sum, max, min, average or a user-defined aggregation function. In addition to losing layer-specific structural information and semantics, additional information of weights are also lost in the case of $\omega$-HoMLN.  For example, if two nodes are strongly connected in one layer and weakly connected or entirely disconnected in another, aggregation conceals these differences and may lead to confusing conclusions about node importance. 

To address this limitation, we adopt the \textit{decoupling-based approach} proposed in \cite{phdThesis/Santra20}, in which each layer of the $\omega$-HoMLN is analyzed independently. Let $G_1, G_2, \dots, G_n$ denote the layers of the $\omega$-HoMLN. Each layer is first analyzed for a specific graph metric (e.g., degree centrality), yielding layer-wise results. These partial results are then combined through a composition function to obtain a final score for nodes within the multilayer structure. This approach preserves the unique contribution of each layer and avoids loss of relationship details. The composition being heuristics-based, comparison with the aggregated ground truth makes sure the composition function preserves accuracy as closely as possible.

One of the goals of this paper is to compute degree centrality-based importance of nodes directly on $\omega$-HoMLNs while maintaining both accuracy and computational efficiency. Finding the most influential or highly connected nodes in a multilayer setting is important in many real-world contexts, such as identifying key researchers across different collaboration topics, or detecting influential individuals in animal interaction networks observed across multiple time periods. 

Contributions of this paper are:
\begin{itemize}
    \item[--] \textbf{Definition} of Degree Centrality for weighted HoMLNs using Boolean OR-aggregation and multiple numerical aggregation functions (sum and max)
    \item[--] \textbf{Algorithms} for directly computing weighted degree centrality nodes of weighted HoMLNs
    \item[--] \textbf{Decoupling-based approach} to preserve structure and semantics of MLNs
    \item[--] \textbf{Heuristics} to improve accuracy, precision, and efficiency of algorithms
    \item[--] \textbf{Experimental analysis} on large number of synthetic, real-world-like, and real-world data sets with diverse graph characteristics
    \item[--] \textbf{Accuracy, Precision, and Efficiency comparisons} with ground truth and naïve approach 
\end{itemize}


The rest of the paper is organized as follows: Section \ref{section:relevant-works} discusses related work. Section \ref{section:decoupling-approach} introduces the decoupling approach used for MLN analysis.  Section \ref{section:weighted-degree} formalizes weighted degree centrality for multilayer networks. Section \ref{section:ground-truth-and-naïve} introduces the definitions, ground truth, naïve approach to weighted-degree centrality. Section \ref{section:max} presents lemmas on max aggregations regarding computation using the decoupling approach, as well as composition computation and results. Section \ref{section:datasets-and-computation} discusses the data sets and computation environments. Section \ref{section:sum} presents lemmas on sum aggregations regarding computation using the decoupling approach, as well as composition computation and results. Section \ref{section:conclusion} discusses the conclusions.

\section{Relevant Works}
\label{section:relevant-works}

Degree centrality is one of the fundamental and widely used measures for identifying influential nodes, due to its simplicity and intuitive interpretation as a measure of a node's local connectivity. In simple graph settings, degree centrality captures local connectivity and has been extensively studied in social network analysis literature  \cite{freeman1978centrality, wasserman1994social, Klein_2010}. However, this formulation assumes unweighted edges and treats all edges as unit weight, thereby failing to capture the intensity or strength of interactions in real-world applications. 

To address this limitation, several extensions of degree centrality have been proposed for weighted graphs. Early work on weighted networks introduced the notion of node strength as the sum of weights associated with the edges incident on a node \cite{barrat2004architecture, Newman2004}. Building on this, \cite{opsahl2010node} proposed a generalized formulation that integrates both degree and strength through a tunable parameter, allowing the centrality to reflect the balance between connectivity and interaction weight. This formulation unifies classical degree centrality and strength-based centrality, allowing flexible modeling across different application domains. Similarly, \cite{YUSTIAWAN2015419} demonstrate the application of the Opsahl method~\cite{opsahl2010node} in social networks, highlighting the importance of interaction weights such as "mentions" and "replies" in identifying influential nodes. Prior work further highlights the importance of weight distribution and interaction intensity in centrality estimations \cite{Candeloro2016, Rachman2013}.

Despite these advancements in weighted graphs (that correspond to a single-layer in an MLN), extending degree centrality to multilayer networks (MLNs) introduces addition challenges. In MLNs, same type nodes may participate in multiple layers representing different types of relationships, and naively extending degree centrality is non-trivial due to the need to account for both inter-layer and intra-layer structure. A direct extension, proposed by \cite{Brodka2011}, defines cross-layer degree centrality based on multilayer neighborhood and aggregated edge weights across layers. While this approach captures  interaction spanning multiple layers, it relies on the aggregation of connections and does not explicitly preserve layer-specific structural properties. 


To address this limitation, \cite{phdThesis/Santra20} proposed a decoupling-based framework for analyzing multilayer networks without transforming them into aggregated single graph representation, thereby preserving the structural and semantic information present as part of MLN layers. The framework analyzes each layer independently and combines results using a composition function. \cite{KDIR2022/PavelSC22} applied the decoupling approach to degree centrality in homogeneous multilayer networks (HoMLNS) without weights, using composition heuristics to approximate rankings derived from aggregated graphs. This framework systematically analyzes the trade-off between computational efficiency and approximation accuracy.

Building on this idea, this paper extends the decoupling-based framework to weighted homogeneous multilayer networks by incorporating edge weights into the layer centrality computation. Unlike aggregation-based methods, this approach preserves intra-layer semantics by computing weighted degree (strength) independently on each layer and combining the results through a composition function. This design enables scalable, interpretable, and weight-aware estimation of degree centrality for large HoMLNs without reliance on application-specific assumptions.  Two kinds of weight aggregation Sum and max are discussed in this paper along with the tunable parameter $\alpha$ set to 1 which gives full emphasis for each edge weight.
\section{Decoupling-Based Approach for MLNs}
\label{section:decoupling-approach}

Multilayer networks (MLNs), also known as multiplexes, consist of several layers of weighted or unweighted simple graphs, where each layer captures a distinct type of interaction among a set of nodes of the same type. In homogeneous MLNs (HoMLNs), node types are same across layers, while edges and weights vary and represent different relationships or time periods of the same system.

MLNs provide a powerful data model for capturing complex systems in which multiple types of interactions must be preserved. Applications of MLNs span a wide range of domains, including modeling connectivity patterns in the brain, analyzing transportation networks, and studying dynamic behavior in social systems.




While HoMLNs preserve richer structural information than attributed graphs, their analysis presents new challenges. Most classical graph algorithm are designed for simple graphs (or single layer of an MLN) and cannot be directly applied to MLNs. A common workaround is to aggregate layers into a single graph using Boolean AND or OR operators for unweighted MLNs. For $\omega$-MLNs, in addition, a function is needed (e.g., Sum, Min, Max, Average, etc.) for aggregating weights. However, such conflation of layers can distort structural patterns and eliminate layer-specific information~\cite{KDIR2022/PavelSC22,BDA/ChakravarthySK19,BDS2023/MukundaRSC}. 


\subsection{Analyzing Multilayer Networks}
Several approaches have been used for computing centrality and other network metrics in homogeneous (HoMLNs), heterogeneous (HeMLNs), and hybrid (HyMLNs) multilayer networks. These approaches vary in how they process the layered structure of MLNs and the extent to which they preserve layer-level semantics. Figure~\ref{fig:mln-approaches} illustrates three main alternatives explored in the literature. 



\begin{itemize}
    \item[--] \underline{Aggregation Approaches:} As shown in Figure~\ref{fig:mln-approaches}(a), MLNs are conflated into a single graph using Boolean (and weight aggregation, if needed) after which standard simple graph algorithms (weighted or unweighted) are applied~\cite{LayerAggDomenicoNAL14, berenstein2016multilayer}. Although computationally convenient, this approach may distort connectivity patterns and obfuscate layer semantics~\cite{phdThesis/Santra20,DBLP:journals/dke/SantraKBC22}.
    \item[--] \underline{Whole-MLN Algorithms:} As shown in Figure~\ref{fig:mln-approaches}(c), algorithms based on random walk methods operate directly over the entire MLN structure, preserving semantics, but incur high computational cost and limited flexibility~\cite{KDIR2022/PavelSC22}.

\end{itemize} 
    

\begin{figure}[!tbh]
    \centering
    \includegraphics[width=3.2in,keepaspectratio]{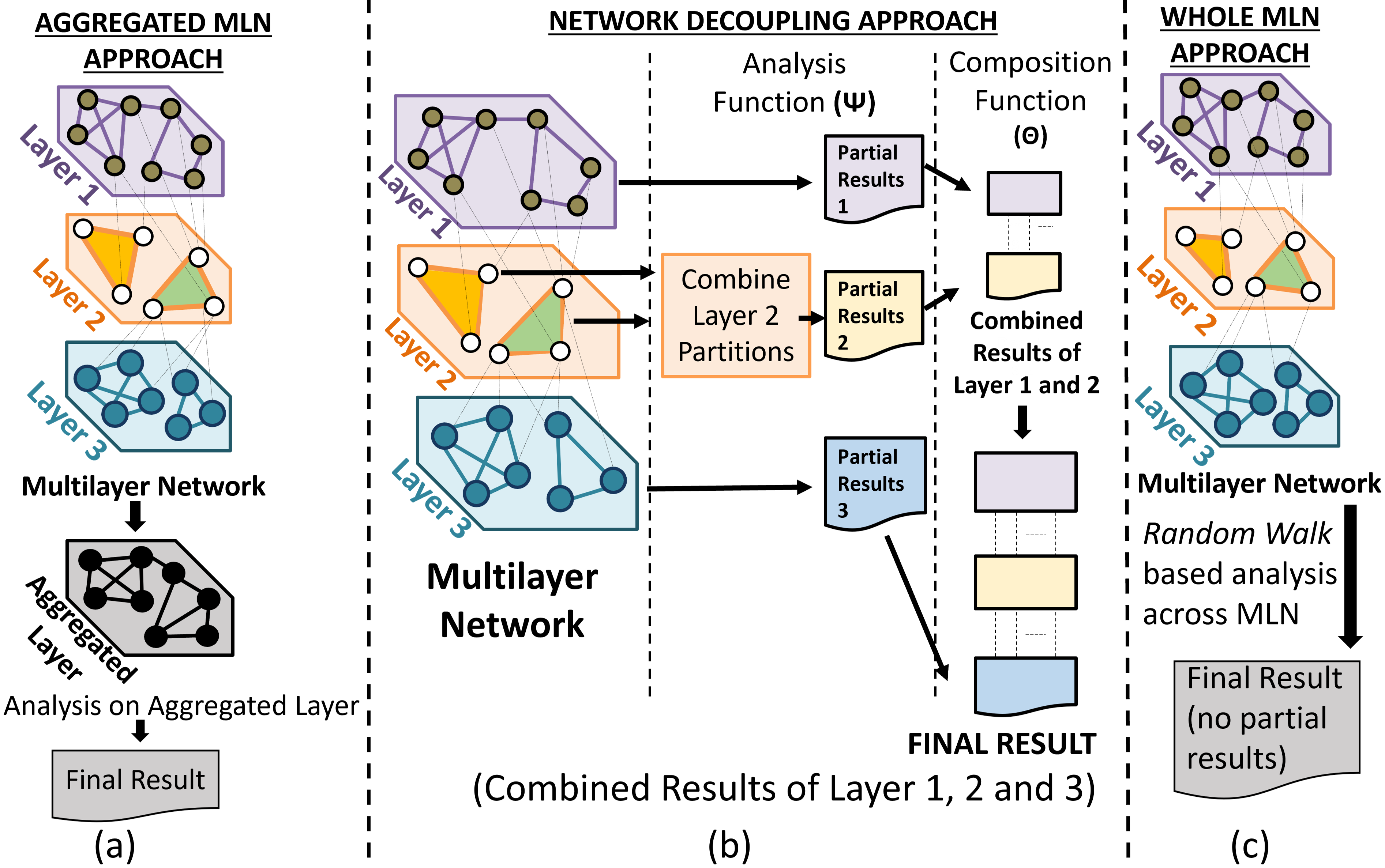}
    \caption{\small (a) Lossy, (b) Decoupling, (c) Whole MLN approaches~\cite{DBLP:journals/dke/SantraKBC22-short}}
   \label{fig:mln-approaches}
\end{figure}

\subsection{The Decoupling Approach}



To address the limitations of aggregation (e.g., loss of semantics) and avoid the cost of whole-MLN methods, the decoupling-based approach was introduced by Santra~\cite{phdThesis/Santra20}. Rather than flattening the network, this framework processes each layer independently using the standard single-layer algorithms and then combines the partial results through a composition function.

Originally developed for unweighted MLNs using Boolean AND/OR aggregation, this approach can also be extended for $\omega$-MLNS, where both structural and weight aggregation choices influence the design of the composition function.

Figure~\ref{fig:decoupling_approach} illustrates the decoupling framework for centrality computation. The general approach for any analysis function $\Psi$ is shown in Figure~\ref{fig:mln-approaches}(b). In the first step, each layer of the HoMLN is analyzed \underline{independently} as a separate graph using the designated analysis function. Since layers in a HoMLN share the same node subset but differ in edge sets and weights, this approach preserves the semantic integrity of each layer. The result of this function is a partial result from each layer. In the second step, the partial results from any two layers are combined using a composition function $\Theta$, which integrates the outcomes to produce results that reflect the multilayer structure as a whole. 




\begin{figure}[!tbh]
    \centering
    \includegraphics[width=\columnwidth,keepaspectratio]{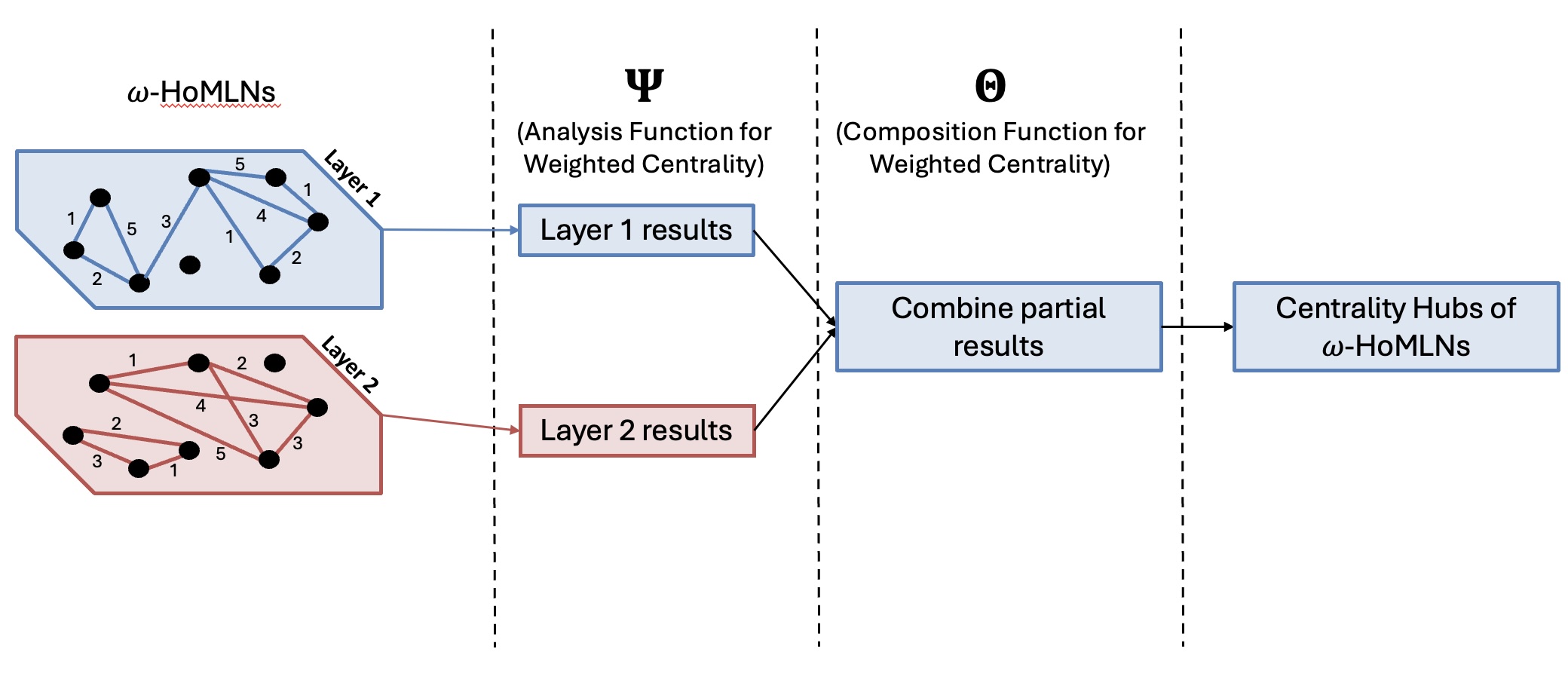}
    \caption{Decoupling approach for Centrality computation in $\omega$-MLNs}
    \label{fig:decoupling_approach}
\end{figure}
The primary challenge in the decoupling-based approach is to identify minimal information to carry over from each layer analysis and developing heuristics based on intuition for the aggregation method used. Reaching ground truth accuracy can be a challenge as well.

\subsubsection{Advantages of the Decoupling Approach}




\begin{itemize}
    \item[--] Supports parallelism, as layers can be processed independently, improving overall computational efficiency.
    \item[--] Preserves layer semantics by avoiding lossy aggregation. 
    \item[--] Utilizes existing single-layer  algorithms
\end{itemize}

\subsubsection{Challenges of the Decoupling Approach}

\begin{itemize}
    \item[--] Due to heuristic-based algorithms necessitated by incomplete information, reaching ground truth accuracy is harder to achieve compared to the conflated single graph approaches. 
    \item[--] Requires new algorithm for composition using heuristics.
    \item[--] Compensating for inter-layer awareness, determining minimal information to be carried over from layers.
\end{itemize}

\section{Degree Centrality for $\omega$-HoMLNs}
\label{section:weighted-degree}

In weighted graphs, degree centrality is commonly extended to not only account for the number of edges incident on a node but also for the strength of those connections. This is particularly important in real-world networks, where edge weights represent interaction intensity, frequency, or cost. However, when such graphs are extended into HoMLNs, computing degree centrality becomes significantly more complex.

As mentioned earlier, recent work on degree centrality in multilayer networks has largely focused on unweighted graphs, using techniques such as Boolean aggregation (e.g., OR or AND operations) and heuristic-based composition. As discussed in Section~\ref{section:relevant-works}, \cite{opsahl2010node} proposed a generalized degree centrality measure that incorporates both the degree and the strength of a node using a tunable parameter \( \alpha \): 

\begin{equation}
    \text{Degree Centrality} (i) = k_i^{1 - \alpha} \times s_i^{\alpha}
    \label{eqn:strength-equation}
\end{equation}
where $k_i$ denotes the (unweighted) degree of node $i$ and $s_i$ denotes its strength, defined as the sum of the weights of edges incident on $i$.

In this work, we place full emphasis on the strength of each interaction so, we set  \(\alpha = 1\) for the formula shown in~\ref{eqn:strength-equation}. Under this setting, the generalized formula reduces to~\ref{eqn:sum-of-strengths}:
\begin{equation}
    {DC}_i = s_i
    \label{eqn:sum-of-strengths}
\end{equation}
where $s_i$ is the sum of the weights of the edges incident at node $i$. By changing the value of $\alpha$, degree and strength contributions of weights can be customized based on application semantics.

This formulation establishes weighted degree centrality in weighted HoMLNs as a strength-based measure that directly reflects the intensity of interactions incident on a node. 
\section{Ground Truth (GT) and Naïve Accuracy}
\label{section:ground-truth-and-naïve}
\subsection{Ground Truth}
To evaluate the accuracy and efficiency of the proposed heuristics, the ground truth must be established. In this paper, for weighted-degree centrality, the ground truth is constructed using Boolean-OR aggregation of all layers, meaning that an edge is included in the aggregated network if it appears in at least one of the layers. For edges that appear in multiple layers, weights are aggregated using one of the following strategies: sum, maximum, minimum, or average. Among these strategies, this paper on weighted-degree centrality focuses on sum and maximum aggregation methods. Note that these functions are separate from the sum function used for strength in Equation~\ref{eqn:strength-equation}.

\begin{itemize}
    \item Sum aggregation captures cumulative interaction strength across layers.
    \item Max aggregation preserves the strongest observed interaction between node pairs. 
\end{itemize}


After the edge weights are aggregated, the strength of each node is computed as the sum of the weights of all edges incident to that node in the aggregated graph as defined in equation~\ref{eqn:sum-of-strengths}. The aggregation graph produced by Boolean-OR combined with sum or max weights is then used to compute degree hubs for the GT weighted graph. Nodes with strengths greater than the graph average are identified as degree hubs and are used as the ground truth for evaluating heuristic performance. Jaccard coefficient, precision, and recall are used for comparing the results with the ground truth.

All heuristics proposed in this paper are tested on two-layer HoMLNs, although the methodology can be generalized to more than two layers. For the decoupling approach, each layer is independently analyzed as outlined earlier in Section~\ref{section:decoupling-approach}. 

\subsection{Naïve Accuracy:} As a baseline for measuring improvement in accuracy and performance of our heuristics, we define a naïve approach for composition using the decoupling-based framework that uses minimal information from each layer.  For each layer, we identify the degree hubs in each individual layer, and take the union (as our aggregation operator is Boolean OR) of the resulting hub sets as the estimated set of hubs for the HoMLN. This method does not utilize any additional information beyond the identified layer hubs, which is the minimal information, and applies a simple heuristic. Naïve degree hub computation for HoMLN also involves minimal amount of computation as part of the composition function.  While this method is computationally inexpensive, the naïve approach is generally not accurate and is unlikely to match the ground truth accuracy except in pathological cases~\cite{phdThesis/Santra20, KDIR2022/PavelSC22}. Therefore, this method establishes a low watermark for accuracy, which can be used to further evaluate the performance of the proposed heuristics.
\section{Weighted Degree Centrality Using Boolean-OR and Max Aggregation}
\label{section:max}

Under max aggregation, the weight of each edge in the ground truth is determined by taking the maximum weight among \textit{overlapping} edges across layers. If an edge appears in both layers, only the larger of the two weights is retained in the aggregated network. 


\subsection{Definitions}
Let the two layers of an HoMLN be denoted as $x$ and $y$, with a shared node set $V$ (also the nodes in the ground truth graph). Let $E_x$ and $E_y$ be edges of layers $x$ and $y$, respectively. For any node $u \in V$, the analysis function computes its layer-wise strength as:

\begin{equation*}
    s_x(u) = \sum_{(u, v) \in E_x} w_x(u, v), \quad
    s_y(u) = \sum_{(u, v) \in E_y} w_y(u, v)
\end{equation*}

\noindent\underline{Ground Truth Using Boolean-OR and Max:} The ground truth graph $G_{GT}$ is formed using \textit{Boolean-OR} on the edge structure and \textit{max} aggregation for the weights on those edges:  
\begin{equation*}
    w_{GT}(u,v) = 
    \begin{cases}
    \max\{w_x(u,v), w_y(u,v)\}, & (u,v) \in E_x \cap E_y \\
    w_x(u,v), & (u,v) \in E_x \setminus E_y \\
    w_y(u,v), & (u,v) \in E_y \setminus E_x \\
    \end{cases}
\label{eqn:max_gt_weights}
\end{equation*}

Using this, the corresponding ground truth strength of node $u$ is:
\begin{equation*}
    s_{GT}(u) = \sum_{v \in V} w_{GT}(u,v)
\end{equation*}

The ground truth graph (or global) average strength is: 
\begin{equation*}
    \bar{s}_{GT} = \frac{1}{|V|} \sum_{u \in V} s_{GT}(u) 
\end{equation*}

A node $u$ is considered a ground truth hub if $s_{GT}(u) > \bar{s}_{GT}$ 

\noindent\underline{Bounds for GT edge strength for max aggregation:} 
Let $G_{GT}$ be the ground truth graph using max aggregation and let $V_{GT}$ and $E_{GT}$ be, respectively,  the set of nodes and edges in the ground truth graph. If we are able to compute the correct strength of an edge in the GT graph using strength information of the edges in the layers, then there is no need to estimate it during composition. However, due to the use of max function for the aggregation this is not possible unless we keep the \textbf{edge strength information for all overlapping edges.} This violates the independent analysis principle behind the decoupling approach and hence the composition function will not have overlap information \underline{even if we keep edge strength information} from each layer! In other words, $s_{GT}(u)$ cannot be computed exactly from $\{s_x(u), s_y(u)\}$ alone because max aggregation depends on edge-level overlap and overlap information is not available during composition\footnote{Note that for overlapping edges, the weights can be different in each layer.}. However, we can compute the upper- and lower-bounds of the GT edge strengths as follows:
\begin{equation*}
s_{GT}^{LB}(u)=\max\{s_x(u), s_y(u)\},
s_{GT}^{UB}(u)=s_x(u) + s_y(u) 
\end{equation*}
where $s_{GT}^{LB}(u)$ corresponds to complete overlap (Lemma~\ref{lemma:max-lemma-b}), and $s_{GT}^{UB}(u)$ corresponds to no overlap, satisfying                                      
\begin{equation*}
    s_{GT}^{LB}(u) \le s_{GT}(u) \le s_{GT}^{UB}(u)
\end{equation*}

Intuitively, $s_{GT}^{LB}(u)$ corresponds to the case of overlap being maximal for the edges incident to $u$, while $s_{GT}^{UB}(u)$ corresponds to the other extreme case of no overlap between the layers. 

\subsection{Lemmas and Intuition}
As Max aggregation of layer edge weights is different from sum used for strength computation, exact computation of strength in the ground truth is not possible. The following lemmas formalize useful monotonicity properties and clarify when a hub can or cannot be inferred from layer-wise hub membership.

\begin{lemma}
    when a ground truth aggregation graph uses max function, the ground-truth node strength $s_{GT}(u)$ of a node $u$ cannot, in general, be determined solely from the layer-wise strengths $\{s_x(u), s_y(u)\}$ unless edge-level overlap information across layers is available.
    \label{lemma:max-lemma-a}
\end{lemma}

\begin{proof}
    Under max aggregation, the strength of a node $u$ in the ground truth graph is 

    \begin{equation*}
        s_{GT}(u) = \sum_{(u,v) \in E_{GT}} \max\{w_x(u,v), w_y(u,v)\}
    \end{equation*}

Evaluating this expression requires knowing, for each edge $(u,v) \in E_{GT}$, whether the edge appears in one layer or in both layers, and, in the latter case, comparing the corresponding weights $w_x(u,v)$ and $w_y(u,v)$ to determine the maximum weight. 

When each layer is processed \textit{independently} (which is the underlying assumption for the decoupling-based framework), the analysis function produces only the aggregate layer-wise strengths $s_x(u)$ and $s_y(u)$. It does not retain information about individual edges, their presence across layers, or information of neighbors between layers. As a result, distinct ground truth edge sets, $E_{GT}$, with different overlap structures and max weights, can influence the same pair of layer-wise strengths $\{s_x(u), s_y(u)\}$ while yielding different values of $s_{GT}(u)$ under max aggregation. 
\end{proof}

\begin{lemma}
    For every node $u \in V_{GT}$, the ground truth strength under max aggregation satisfies 
    \begin{equation*}
        s_{GT}(u) \ge \max\{s_x(u), s_y(u)\}.
    \end{equation*}
    \label{lemma:max-lemma-b}
\end{lemma}
\begin{proof}
    Consider node $u$ and any incident edge $(u,v) \in E_x$. In the ground truth graph constructed using Boolean-OR and max aggregation, the contribution of this edge to $s_{GT}(u)$ is:
    \begin{equation*}
        w_{GT}(u,v) = 
        \begin{cases}
            w_x(u,v), & (u,v) \notin E_y, \\
            
            \max\{w_x(u,v), w_y(u,v)\}, & (u,v) \in E_x \cap E_y \\
        \end{cases}
        \label{eqn:max-proof-b}
    \end{equation*}

    In both cases, $w_{GT}(u,v) \ge w_x(u,v)$. Summing over all edges incident to $u$ in $Layer_x$ results in:
    \begin{equation*}
        s_{GT}(u) \ge \sum_{(u,v) \in E_x} w_x(u,v) = s_x(u)
    \end{equation*}

    By symmetry, the same argument holds for $Layer_y$, implying $s_{GT}(u) \ge s_y(u)$. Therefore, 
    \begin{equation*}
        s_{GT}(u) \ge \max\{s_x(u), s_y(u)\}
    \end{equation*}
\end{proof}

\begin{lemma}
    Under the Boolean-OR composition and max aggregation, non-hub status in every layer does not guarantee non-hub status in the ground-truth aggregation graph. 
    \label{lemma:max-lemma-c}
\end{lemma}

\begin{proof}
    Hub membership is determined by comparing node strength to the average node strength of the corresponding graph. Consider a node $u$ such that 

    \begin{equation*}
        s_x(u) \le \bar{s}_x
        \quad \text{and} \quad
        s_y(u) \le \bar{s}_y
    \end{equation*}

    so that $u$ is not a hub in either layer. Under the Boolean-OR composition, the ground-truth graph may contain edges incident to $u$ that originate from different layers. By Lemma \ref{lemma:max-lemma-b}, the contribution of each edge to the ground-truth strength is never less than its contribution in either of the individual layer. Consequently, the ground-truth strength of $u$ may exceed both $s_x(u)$ and $s_y(u)$

    Since hub membership in the ground-truth graph is evaluated with respect to the average node strength $\bar{s}_{GT}$ rather than the layer averages, the relationship

    \begin{equation*}
        s_x(u) \le \bar{s}_x
        \quad \text{and} \quad
        s_y(u) \le \bar{s}_y
    \end{equation*}
    does not imply 
    \begin{equation*}
        s_{GT}(u) \le \bar{s}_{GT}
    \end{equation*}
    Therefore, non-hub status in each layer does not guarantee non-hub status in the ground-truth aggregated graph. 
\end{proof}

Lemmas \ref{lemma:max-lemma-a}, \ref{lemma:max-lemma-b}, and \ref{lemma:max-lemma-c} together demonstrate the limitations of inferring hub membership under the max-aggregation. Lemma \ref{lemma:max-lemma-a} shows that the ground-truth strength of a node cannot, in general, be determined solely from the layer-wise strengths {$s_x(u), s_y(u)$} without edge-level overlap information. Lemma \ref{lemma:max-lemma-b} shows that the ground truth is bounded below by the larger of the layer strengths. While this monotonicity property guarantees that $s_{GT}(u)$ cannot be less than max $(s_x(u), s_y(u))$, it does not determine hub membership in the aggregated graph. Lemma \ref{lemma:max-lemma-c} further demonstrates that non-hub status in every layer does not necessarily transfer to the aggregated graph. Hub classification depends on whether $s_{GT}(u)$ exceeds the average node strength $\bar{s}_{GT}$, and max aggregation can increase the strengths of other nodes and thereby raise the global average. Consequently, layer hub membership does not necessarily transfer to the aggregated network. The following example illustrates how hub status can change under the max aggregation even when layer strengths are known.

\begin{figure}[h!]
    \centering

    \begin{minipage}{0.4\columnwidth}
    \centering
    
    \includegraphics[width=\columnwidth]{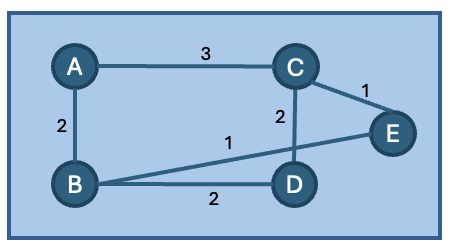}\\[0.05em]
    \textbf{Layer\textsubscript{X}}
    \end{minipage}
    \hfill
    \begin{minipage}{0.4\columnwidth}
    \centering
    \includegraphics[width=\columnwidth]{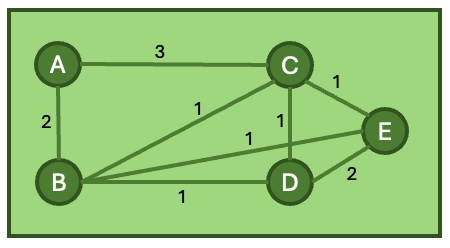}\\[0.05em] 
    \textbf{Layer\textsubscript{Y}}
    \end{minipage}
    \caption{Layer-wise weighted edges in the multilayer network.}
    \label{fig:max-example-for-hubs-xy}

    \begin{minipage}{0.4\columnwidth}
    \centering
    \includegraphics[width=\columnwidth]{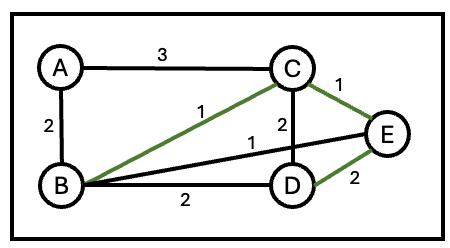} \\[0.05em]
    \end{minipage}

    \caption{Ground truth network ($G_{GT}$) using Boolean-OR and Max aggregation.}
    \label{fig:max-example-for-hubs-gt}
\end{figure}

\begin{table}[h!]
    \caption{Node strengths in $Layer_X$, $Layer_Y$, and the ground truth $GT$}
    \label{tab:strengths-e2}
    \begin{tabularx}{\columnwidth}{|c|X|X|X|}
    \hline
    \textbf{Node} & \textbf{$Layer_X$} & \textbf{$Layer_Y$} & \textbf{GT} \\
    \hline
        A & 5 & 5 & 5 \\ \hline
        B & 5 & 5 & 6 \\ \hline
        C & 6 & 6 & 7 \\ \hline
        D & 4 & 4 & 6 \\ \hline
        E & 2 & 4 & 4 \\ \hline
        average & 4.4 & 4.8 & 5.6 \\
    \hline
    \end{tabularx}
\end{table}

In this example, node $A$ is a hub in both $Layer_X$ and $Layer_Y$ since 
    \begin{equation*}
        s_X(A) = 5 > \bar{s}_X=4.4, \quad s_Y(A) = 5 > \bar{s}_Y=4.8.
    \end{equation*}

However, in the ground-truth graph, 
    \begin{equation*}
        s_{GT}(A) = 5 < \bar{s}_{GT}=5.6,
    \end{equation*}

so node $A$ is not a hub in the ground truth graph. This shows that hub membership in both layers does not imply hub membership in the aggregated graph. Conversely, node $D$ is not a hub in either layer since
    \begin{equation*}
        s_X(D) = 4 > \bar{s}_X=4.4, \quad s_Y(D) = 4 > \bar{s}_Y=4.8.
    \end{equation*}
Yet, in the ground-truth graph, 
    \begin{equation*}
        s_{GT}(D) = 6 < \bar{s}_{GT}=5.6,
    \end{equation*}
So node D becomes a  hub in $\bar{s}_{GT}$. So, non-hub status in both layers does not guarantee non-hub status after max aggregation.

Therefore, node $A$ and $D$ show that max aggregation does not preserve hub membership. A node may lose hub status after aggregation, while another may gain hub status, even when layer-wise strengths are known. 

\subsection{Composition Strategies for Max Aggregation} 
Since exact max-aggregation strength requires edge-level overlap, composition, in the absence of that information, has to use  heuristic for estimates. 

\underline{(1) Lower bound (LB) Heuristics} The composition strength estimate is $\hat{s}_{LB}(u) = \max\{s_x(u), s_y(u)\}$, which is conservative under overlap and avoids over-counting. This strategy can miss true hubs when contributions from distinct neighbors across layers accumulate in $G_{GT}$.

\underline{(2) Upper bound (UB) Heuristic} The composition strength estimate is $\hat{s}_{UB}(u) = s_x(u) + s_y(u)$, which is optimistic and ensures potentially strong nodes are not excluded. While this may overestimate $s_{GT}(u)$ when overlap is high, it is often effective for hub recovery because it prioritizes recall, minimizing the risk of missing true ground-truth hubs. 

We first describe the data sets used in this paper. We describe synthetic data sets, along with generation of layers, creation of aggregate graph for GT. We also indicate the real-world data sets used for experimental analysis and heuristic validation.

\section{Data Sets and Computation Environments}
\label{section:datasets-and-computation}
Synthetic and real-world datasets are used in this work to validate the accuracy, efficiency, and scalability of algorithms proposed for both sum and max aggregation degree centrality of $\omega$-MLNs. They are also validated for diverse graph/MLN characteristics (for synthetic datasets), including weights, to make sure heuristics-based algorithms perform well across diverse graph/MLN characteristics. We also use real-world data sets to assert our heuristics. However, We do not have control over real-world dataset characteristics. Hence, both real-world and MLNs with diverse graph characteristics are tested to make sure our proposed algorithms works across varied datasets.  Synthetic datasets were generated from base graphs generated using \textbf{subgen}~\cite{subgen} (from the AI Lab at Washington State University) and \textbf{PaRMAT}~\cite{wsvr} as described in the work by Pavel~\cite{KDIR2022/PavelSC22} and Mukunda~\cite{msThesis/Mukunda21}.

These  datasets were extended to add weights randomly in a given range to generate weighted graphs. These weighted graphs were then converted into layers by specifying: number of layers, edge split or distribution among layers, and overlap as a percentage across layers. Note that all layers contain the same set of nodes\footnote{This is not a requirement of HoMLNs. Having non-identical sets generates disconnected graphs in ground truth. Our heuristics works for them as well.}.  This approach allows fine-grained control over structural properties such as edge distribution between layers, percentage of edge overlap, and weight ranges. This makes them suitable for benchmarking under a variety of multilayer configurations. Real-world datasets~\cite{collier_2024_14168300, hu2020ogb} were obtained from domain-specific studies, where each layer represents observations from different time points. 

Once the layers are generated or identified (in case of real-world datasets), ground truth graph need to be generated from the layers using a specified Boolean operator and the function for aggregating the strengths from layers for the ground truth. We have chosen sum and max function for this paper. Other possible functions are minimum, average, and custom combining of weights. \textit{Note that the definition for degree centrality used is the same as the one used for the ground truth. Equation~\ref{eqn:sum-of-strengths} is used for identifying degree centrality nodes in both simple graphs and MLNs}.

Analysis on the ground truth graph is performed using existing algorithms available for weighted simple graphs.  In contrast, analysis output (and any additional carryover information appropriate for the heuristics being used) of the layers is passed on to the composition function which applies the chosen heuristics for generating the metric for MLNs.  

We first outline the two methods used for generating synthetic HoMLN datasets and also describe the real-world datasets used in the evaluation. 
Two general strategies were considered for generating synthetic homogeneous multilayer networks from base unweighted graphs:

\begin{itemize}
    \item Bottom-up: Combining two existing weighted graphs into a simple graph using a Boolean operator (AND or OR) and a weight aggregation function (one of sum, minimum, maximum, average.)
    \item Top-down: Start from a simple unweighted graph and partition it into layers as needed including adding weights randomly using a range. This method was used to generate all synthetic MLN datasets in this work from the base synthetic graph generated by Subgen. 
\end{itemize}

\subsection{Synthetic $\omega$-MLN Generation}

A top-down approach is used for this. The reason for this choice is that it affords much finer control in the generation of layers with diverse characteristics. For our purpose, a number of layer characteristics are important to make sure that the heuristics developed for composition works well across MLNs with varied layer characteristics. Subgen~\cite{subgen} is used to generate the base graph  from which layers are generated. Parameters used for layer generation include: edge distribution (e.g., 90-10, 70-30, and 50-50), weight range (e.g., randomly chosen in the interval [1, 10]), and edge overlap across layers (e.g., 0\%, 25\%, 50\% 75\%, and 100\%) to make sure a wide range of MLN characteristics are tested.

Our layer generation starts with a single unweighted base graph and generates a k-layer HoMLN by distributing the edges across layers using parameters indicated above  and given in a configuration file. The following are user defined parameters that control the generation of an MLN from the base graph:

\begin{itemize}
    \item \underline{Number of layers:} Specifies how many layers the base graph is to be split into. For this work, we have used two.
    \item \underline{Edge Distribution:} Determines the percentage of total edges allocated to each layer (e.g., 50-50, 70-30, 90-10). This is to simulate different graph densities among layers and explore their impact on the  performance of the heuristic.
    \item \underline{Aggregation Method:} Specifies the function to be used when computing the ground truth from the two layers (e.g. sum, max, min, average). This ensures consistency in the weights when the layers are later recombined.
    \item \underline{Weight Range:} Specifies the minimum and maximum edge weights to assign during generation. For weighted graphs, edge weights are randomly chosen from the specified range (min, max). In this paper, we use the range (1,10). For unweighted graphs, all edge weights are set to 1 by default.
    \item \underline{Random seed:} A random seed is used to ensures reproducibility of MLNs from the base graph. 
\end{itemize}

It is important to understand that once the layers are generated from the base graph, we still need to generate the aggregated/flattened ground truth graph \textbf{after} the generation of layers as the initial base graph does not have any weights.
For this work, each generated HoMLN consists of two layers with shared node sets but different edge distributions and weights. These layers are then aggregated using the selected aggregate function (e.g., sum, max, or min) and the Boolean operator (OR in our case) for conflating layer graphs into a ground truth graph for evaluation.

\subsubsection{Characteristics of synthetic datasets used} 
The datasets were generated from larger unweighted base graphs using the top-down approach. For each base graph, the total edge set is split into two layers using different edge distribution configurations (90-10,70-30, and 50-50), with edge overlap percentages of 0\%, 25\%, 50\%, and 75\% being applied. In table ~\ref{table:dataset_50-50}:

\begin{itemize}
    \item The \textbf{dataset} column indicates the dataset name.
    \item \textbf{Edges ($l_1$)} and \textbf{Edges ($l_2$)} specify the number of edges in layer 1 and layer 2, respectively.
    
    
    \item \textbf{Max Deg} gives maximum node degrees across both layers. The minimum degree for all nodes shown in the datasets are 0. 
    
    \item \textbf{Min Wt} and \textbf{Max Wt} indicate the range of edge weights assigned randomly within the specified interval.
    
    \item \textbf{\% Overlap} represents the percentage of edges shared across both layers.
    
    \item \textbf{\% Disjoint} indicates the percentage of edges that are unique to each layer. 
\end{itemize}

\begin{table}[!htbp]
\caption{Characteristics of Synthetic Graphs for a 50-50 Distribution Split Using Sum Aggregation}
\label{table:dataset_50-50}
\resizebox{\columnwidth}{!}{%
\begin{tabular}{| c | c | c | c | c | c | c |}
    \hline

\textbf{Dataset (with \# }   & \textbf{Nodes}           & \textbf{Edges} & \textbf{Edges } & \textbf{Max}  & \textbf{\% } & \textbf{\% }  \\ \textbf{of nodes and edges)}& & \textbf{($l_1$)}&\textbf{($l_2$)}&\textbf{Deg}&\textbf{Overlap}&\textbf{Disjoint} \\ \hline
\textbf{100KV2ME}            & \textbf{100000}  & 1000293          & 999707           & 2049         & 0                & 100                \\ 
\hline
                             &                  & 1249831          & 1250111          & 2560         & 24.997           & 75.003             \\ 
\hline
                             &                  & 1500107          & 1499570          & 3074         & 49.984           & 50.016             \\ 
\hline
                             &                  & 1749729          & 1749342          & 3569         & 74.954           & 25.046             \\ 
\hline
                             &                  & 2000000          & 2000000          & 4041         & 100              & 0                  \\ 
\hline
\textbf{100KV5ME}            & \textbf{100000}  & 2499458          & 2500542          & 4442         & 0                & 100                \\ 
\hline
                             &                  & 3123569          & 3125641          & 5502         & 24.984           & 75.016             \\ 
\hline
                             &                  & 3748904          & 3749458          & 6615         & 49.967           & 50.033             \\ 
\hline
                             &                  & 4373931          & 4373923          & 7667         & 74.957           & 25.043             \\ 
\hline
                             &                  & 5000000          & 5000000          & 8747         & 100              & 0                  \\ 
\hline
\textbf{200KV1ME}            & \textbf{200000 } & 500482           & 499518           & 716          & 0                & 100                \\ 
\hline
                             &                  & 625021           & 624979           & 910          & 25               & 75                 \\ 
\hline
                             &                  & 749857           & 750143           & 1086         & 50               & 50                 \\ 
\hline
                             &                  & 874897           & 874477           & 1256         & 74.937           & 25.063             \\ 
\hline
                             &                  & 1000000          & 1000000          & 1423         & 100              & 0                  \\ 
\hline
\textbf{200KV10ME}           & \textbf{200000 } & 4997753          & 5002247          & 6305         & 0                & 100                \\ 
\hline
                             &                  & 6248645          & 6251158          & 7835         & 24.998           & 75.002             \\ 
\hline
                             &                  & 7499737          & 7500168          & 9432         & 49.999           & 50.001             \\ 
\hline
                             &                  & 8749628          & 8748806          & 10955        & 74.984           & 25.016             \\ 
\hline
                             &                  & 10000000         & 10000000         & 12486        & 100              & 0                  \\ 
\hline
\textbf{1MV8ME}              & \textbf{1382908} & 4230583          & 4234730          & 5619         & 0                & 100                \\ 
\hline
                             &                  & 5289557          & 5291924          & 7038         & 24.998           & 75.002             \\ 
\hline
                             &                  & 6348208          & 6348440          & 8409         & 49.984           & 50.016             \\ 
\hline
                             &                  & 7406680          & 7406196          & 9820         & 74.983           & 25.017             \\ 
\hline
                             &                  & 8465313          & 8465313          & 11219        & 100              & 0                  \\
\hline
		\end{tabular}%
        }
		
\end{table}

\subsection{Real-world Datasets}

To complement the synthetic datasets, real-world homogeneous multilayer network (HoMLNs) have ben  used to validate the effectiveness of the proposed heuristics for computing weighted degree centrality. These datasets capture real interactions from two different domains and provide layered views of the same set of entities over time.

\begin{itemize}
    \item \textbf{Co-Authorship Network:} This dataset is derived from the Open Graph Benchmark’s ogbl-collab dataset~\cite{hu2020ogb}. This network represents a subset of collaborations between authors (Co-authorship) which was indexed by Microsoft Academic Graph (MAG). Each node represents an author, and an edge between two authors indicates a co-authored paper. Each edge is timestamped by year and weighted according to the number of joint publications for a particular year. For this paper, two layers of this dataset were extracted to form a two-layer HoMLN:
    \begin{itemize}
        \item Layer 1: Co-authorship's from the year 2004
        \item Layer 2: Co-authorship's from the year 2005
    \end{itemize}
    The nodes remain consistent across both layers, while the edge sets and weights differ depending on the collaboration activity for each year. The edge weights were retained as-is from the dataset to reflect the intensity of the interaction.
    
    Table~\ref{table:Collab_Dataset} shows the characteristics of the each layer for the Co-Authorship networks



\begin{table}[!ht]
    \caption{Structural properties of the Co-Authorship dataset across two layers (2004 and 2005)}
    \label{table:Collab_Dataset}
    \resizebox{\columnwidth}{!}{%
    
    \begin{tabular}{| c | c | c |}
        \hline
        \textbf{Property} & \textbf{Co-authorship } & \textbf{Co-authorship } \\ 
         &\textbf{2004} & \textbf{2005} \\
        \hline
        Number of Nodes & 235{,}868 & 235{,}868 \\
        \hline
        Number of Edges & 32{,}880 & 42{,}257 \\
        \hline
        Network Density & \(1.18 \times 10^{-6}\) & \(1.52 \times 10^{-6}\) \\
        \hline
        Number of  & 222{,}422 & 219{,}924 \\
        Connected Components & & \\
        \hline
        Minimum node degree & 0 & 0 \\
        \hline
        Maximum node degree & 47 & 56 \\
        \hline
        Minimum edge weight & 1 & 1 \\
        \hline
        Maximum edge weight & 7 & 7 \\
        \hline
        Overlapping edges & 6.58\% & 6.58\% \\
        \hline
        Disjoint edges & 93.42\% & 93.42\% \\
        \hline
    \end{tabular}
    }

\end{table}


    \item \textbf{Ant Interaction Network:} This dataset is drawn from a curated repository of animal social networks that compiles over 1,000 networks from diverse species~\cite{collier_2024_14168300}. One of the networks from this repository captures interactions in an ant colony over multiple days. Each node represents an individual ant, and edges reflect observed interactions between the ants. The edges are defined by spatial proximity, meaning a pair of ants is said to have interacted when the front end of one ant enters the trapezoidal proximity zone of another ant. For this paper, two observations days were selected:
    \begin{itemize}
        \item Layer 1: Interactions observed on Day 1
        \item Layer 2: Interactions observed on Day 3
    \end{itemize}
    \end{itemize}
    Edge weights indicate the frequency of interactions during the respective day. Since the colony is the same, the node set remains the same across both layers.

    Table \ref{table:Ants_Dataset} shows the characteristics of the each layer for the Ant Interaction networks




\begin{table}[!ht]
    \caption{Structural properties of the Ants interaction dataset across two layers (Days 1 and 3)}
    \label{table:Ants_Dataset}
    \resizebox{\columnwidth}{!}{%
    \begin{tabular}{| c | c | c |}
        \hline
        \textbf{Property} & \textbf{Ants Day 01} & \textbf{Ants Day 03} \\
        \hline
        Number of nodes & 164 & 164 \\
        \hline
        Number of edges & 10{,}731 & 10{,}090 \\
        \hline
        Network density & 0.8029 & 0.7549 \\
        \hline
        Number of connected components & 1 & 1 \\
        \hline
        Minimum node degree & 41 & 20 \\
        \hline
        Maximum node degree & 160 & 158 \\
        \hline
        Minimum edge weight & 1 & 1 \\
        \hline
        Maximum edge weight & 229 & 108 \\
        \hline
        Overlapping edges & 74.92\% & 74.92\% \\
        \hline
        Disjoint edges & 25.07\% & 25.07\% \\
        \hline
    \end{tabular}}
    
\end{table}

We present two aggregation types (sum and max) for generating ground truth graph and for both we adopt a decoupling-based framework consisting of (i) a layer-based analysis function and (ii) a heuristic-based composition function. Degree centrality definition remains the same as shown in Equation~\ref{eqn:sum-of-strengths}.
\section{Experimentation Results for Max composition}
We present experimental results using the data sets described in section \ref{section:datasets-and-computation} for a family of algorithms that use the max heuristics described in section \ref{section:max} for the composition function.
Accuracy is measured using Jaccard similarity against the ground truth hub set, while runtime captures the cost of analysis, composition, and ground truth aggregation.

\begin{figure}
    \centering
    \includegraphics[width=\columnwidth]{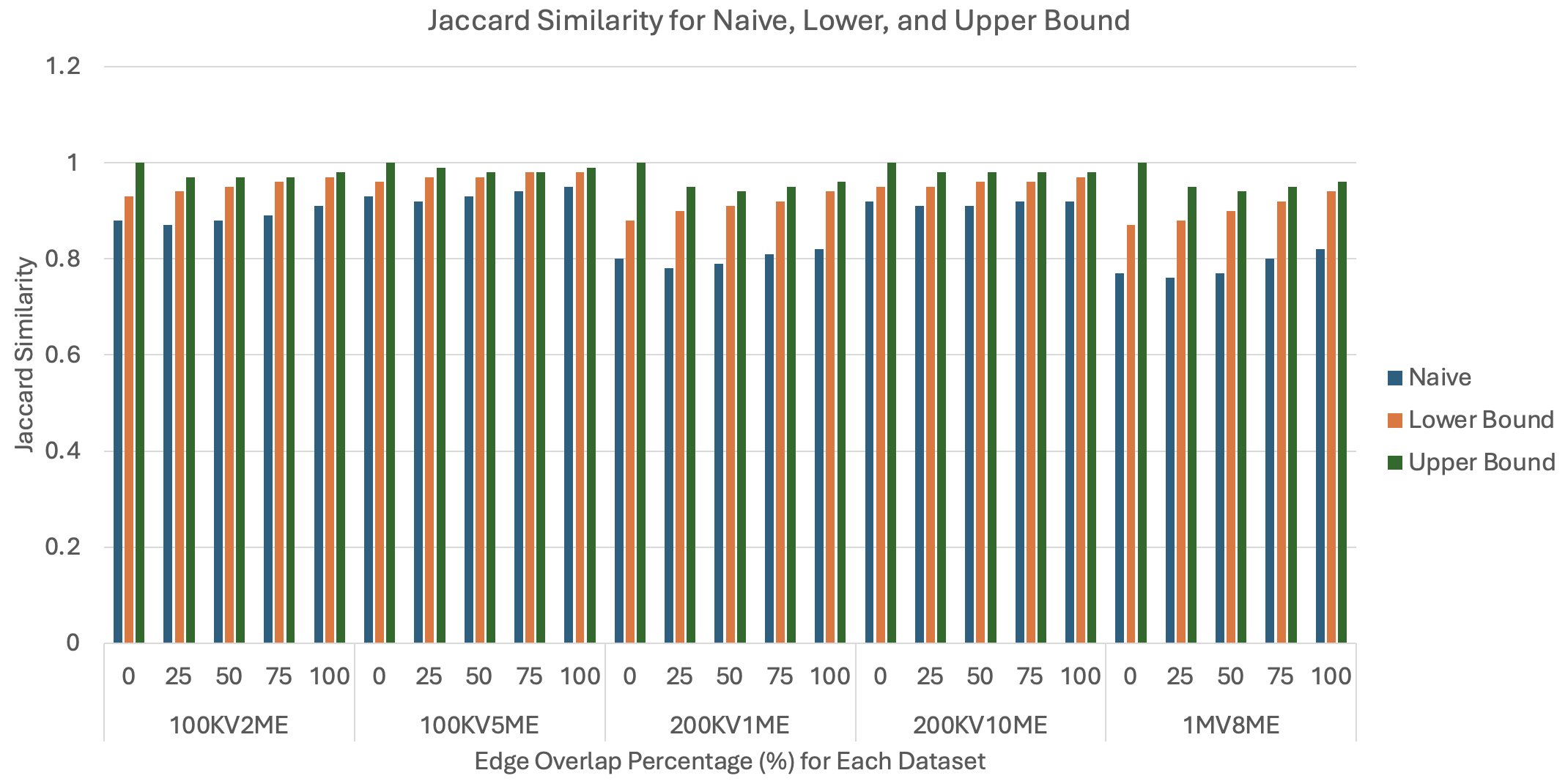}
    \caption{Accuracy Comparison Using Jaccard Similarity for Naive, Lower Bound and Upper Bound Strategies}
    \label{fig:max_jaccard-syn}
\end{figure}

\underline{Accuracy and Hub Identification:}
Figure \ref{fig:max_jaccard-syn} presents the Jaccard similarity achieved by the Naïve, Lower Bound (LB), and Upper Bound (UB) heuristic across the synthetic datasets and overlap configurations. Across the datasets, the UB heuristic consistently achieves the highest Jaccard similarity, followed by the LB heuristic and then the Naïve approach. This behavior follows the construction of the bounds. The UB heuristic estimates the composed strength of a node as $(s_x(u) + s_y(u)$, while the LB heuristic estimates it as $(\max(s_x(u), s_y(u)))$. Since the ground-truth strength under max aggregation is bounded by these, the UB is less likely to exclude nodes that may become hubs after composition. In contrast, the LB heuristic is more conservative and may underestimate nodes whose aggregate importance arises from contributions distributed across the layers. Consequently, the LB heuristic improves upon the Naïve strategy but generally remains less accurate than the UB heuristic. 

The Naïve strategy consistently exhibits the lowest Jaccard similarity because it does not estimate aggregate node strength. Instead, it relies solely on the layer-wise hub membership to infer hub status in the aggregated graph. However, as mentioned by Lemma \ref{lemma:max-lemma-b}, hub membership under max aggregation cannot, in general, be determined solely from layer-wise hub membership. Nodes that are not classified as hubs in either layers may still emerge as hubs in the aggregated graph due to the interaction between edge weights and overlap patterns. As a result, the Naïve strategy frequently excludes relevant node candidates and achieves lower similarity with the ground-truth than both bound-based heuristics. 

\begin{table}
\centering
\caption{Layer Density of the Synthetic Datasets for a 50-50 Edge Distribution Split}
\label{table:dataset_50-50_density}
\resizebox{\columnwidth}{!}{
\begin{tabular}{|c|c|c|c|c|} 
\hline
\textbf{Dataset (with \# }   & \textbf{Nodes}   & \textbf{\% }     & \textbf{Density} & \textbf{Density}  \\
\textbf{of nodes and edges)} &                  & \textbf{Overlap} & \textbf{($l_1$)} & \textbf{($l_2$)}  \\ 
\hline
\textbf{100KV2ME}            & \textbf{100000}  & 0                & 0.00020006       & 0.00019994        \\ 
\hline
                             &                  & 24.997           & 0.00024997       & 0.00025002        \\ 
\hline
                             &                  & 49.984           & 0.00030002       & 0.00029992        \\ 
\hline
                             &                  & 74.954           & 0.00034995       & 0.00034987        \\ 
\hline
                             &                  & 100              & 0.0004           & 0.0004            \\ 
\hline
\textbf{100KV5ME}            & \textbf{100000}  & 0                & 0.00049990       & 0.00050011        \\ 
\hline
                             &                  & 24.984           & 0.00062472       & 0.00062513        \\ 
\hline
                             &                  & 49.967           & 0.00074979       & 0.00074990        \\ 
\hline
                             &                  & 74.957           & 0.00087479       & 0.00087479        \\ 
\hline
                             &                  & 100              & 0.00100001       & 0.00100001        \\ 
\hline
\textbf{200KV1ME}            & \textbf{200000 } & 0                & 0.00002502       & 0.00002498        \\ 
\hline
                             &                  & 25               & 0.00003125       & 0.00003125        \\ 
\hline
                             &                  & 50               & 0.00003749       & 0.00003751        \\ 
\hline
                             &                  & 74.937           & 0.00004375       & 0.00004372        \\ 
\hline
                             &                  & 100              & 0.00005          & 0.00005           \\ 
\hline
\textbf{200KV10ME}           & \textbf{200000 } & 0                & 0.00024989       & 0.00025011        \\ 
\hline
                             &                  & 24.998           & 0.00031243       & 0.00031256        \\ 
\hline
                             &                  & 49.999           & 0.00037499       & 0.00037501        \\ 
\hline
                             &                  & 74.984           & 0.00043748       & 0.00043744        \\ 
\hline
                             &                  & 100              & 0.0005           & 0.0005            \\ 
\hline
\textbf{1MV8ME}              & \textbf{1382908} & 0                & 0.00000442       & 0.00000443        \\ 
\hline
                             &                  & 24.998           & 0.00000553       & 0.00000553        \\ 
\hline
                             &                  & 49.984           & 0.00000664       & 0.00000664        \\ 
\hline
                             &                  & 74.983           & 0.00000775       & 0.00000775        \\ 
\hline
                             &                  & 100              & 0.00000885       & 0.00000885        \\
\hline
\end{tabular}
}
\end{table}

Figure \ref{fig:max_jaccard-syn} also suggests that graph density exerts an influence on accuracy. The synthetic datasets span nearly three orders of magnitude in density, ranging from $4.42\times 10^{-6}$ for 1MV8ME to $1.00\times 10^{-3}$ for 100KV5ME. The sparsest datasets, 200KV1ME and 1MV8ME, consistently exhibit lower Jaccard similarity than the denser datasets, 100KV5ME and 200KV10ME. This can be explained by the effect of density on node strength. In sparse graphs, node strength is determined by relatively few edges, making hub rankings more sensitive to approximation errors. In contrast, denser graphs distribute strength across many edges, which reduces the influence of individual edge-level errors and producing more stable hub rankings. As a result, the heuristics achieve a higher Jaccard similarity score on denser datasets. The effect of density is particularly evident for the Naïve strategy, where the performance gap relative to the UB increases substantially on sparse datasets. This observation suggests that choice of heuristic is not the primary source of accuracy loss, but the exclusion of strength information. Both UB and LB preserve quantitative layer-wise strength information, whereas the Naïve approach relies exclusively on hub membership. Overall, the results indicate that preserving strength information during composition is important when identifying hubs under max aggregation. 

\begin{figure}
    \centering
    \includegraphics[width=\columnwidth]{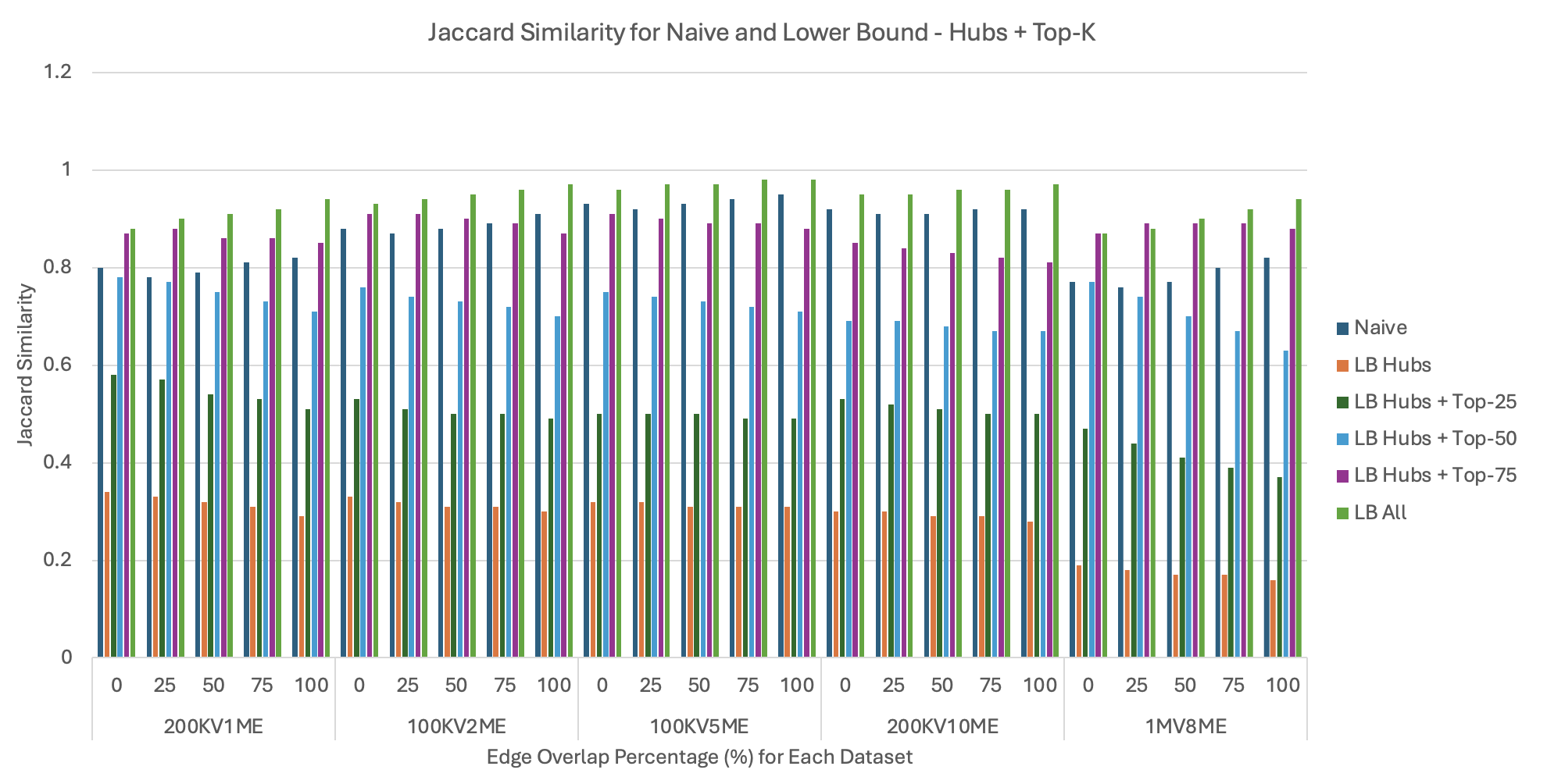}
    \caption{Jaccard similarity for LB Hubs and Hubs + Top-$k$ across datasets and overlap levels}
    \label{fig:max_jaccard-syn-lb-hubs-topk}
\end{figure}

\begin{figure}
    \centering
    \includegraphics[width=\columnwidth]{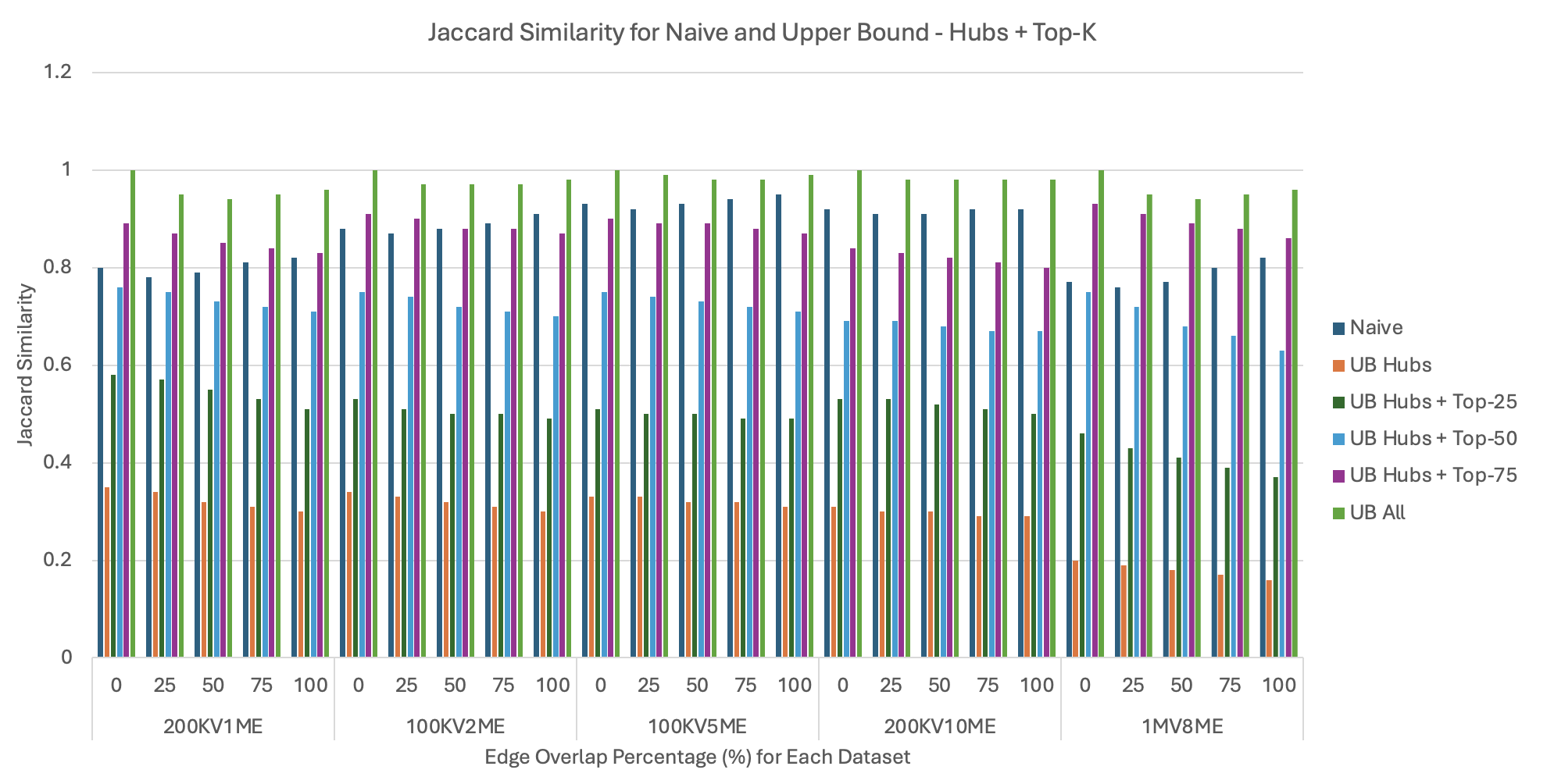}
    \caption{Jaccard similarity for UB Hubs and Hubs + Top-$k$ across datasets and overlap levels}
    \label{fig:max_jaccard-syn-ub-hubs-topk}
\end{figure}

The effect of hub-candidate selection is shown in Figures \ref{fig:max_jaccard-syn-lb-hubs-topk} and \ref{fig:max_jaccard-syn-ub-hubs-topk}. Across all datasets and overlap configurations, the Hubs-Only strategy consistently achieves the lowest Jaccard similarity, often performing substantially worse than the Naïve baseline. This results indicates that restricting composition exclusively to layer-wise hub nodes is overly selective under max aggregation and frequently excludes nodes that become hubs in the aggregated graph. This observation is consistent with Lemma \ref{lemma:max-lemma-b}, which establishes that hub membership under max aggregation cannot, in general, be inferred solely from layer-wise hub membership. 

Expanding the candidate set through Top-$K$ strategies produces a substantial improvement in accuracy for both LB and UB formulations. The dramatic increase in Jaccard similarity when moving from Hubs-Only to Top-$k$ strategies suggests that the primary source of error is not the bound approximation itself, but instead the premature elimination of candidate nodes. By adding a relatively small set of high-strength nodes beyond the strict hub set, the Top-$k$ strategy is able to recover the more of the ground-truth hubs that would have been otherwise excluded.

For the LB formulation, accuracy generally increases as the candidate set expands. However, a noticeable gap between the Top-$k$ strategies and the All strategy remains across most datasets. This behavior reflects the conservative nature of the LB approximation, which can still exclude relevant candidates even after set expansion. In several configurations, Naïve performs better than the Top-$k$ as candidate filtering can remove nodes from composition even though they may become important in the aggregated graph. 

A similar trend can be observed for the UB formulation. However, the difference between the Top-$k$ strategies and the All strategy are considerably smaller since the UB heuristic preserves a larger set of potential hub candidates through its additive strength estimate. Overall, Figures \ref{fig:max_jaccard-syn-lb-hubs-topk} and \ref{fig:max_jaccard-syn-ub-hubs-topk} indicate that candidate-set selection has an influence on accuracy for Top-$k$ strategy. While UB consistently achieves the highest accuracy, the largest performance differences arise from how aggressively nodes are filtered before composition. Preserving a sufficiently large candidate set appears to be more important for accuracy hub recovery. 

\begin{figure}
    \includegraphics[width=\columnwidth]{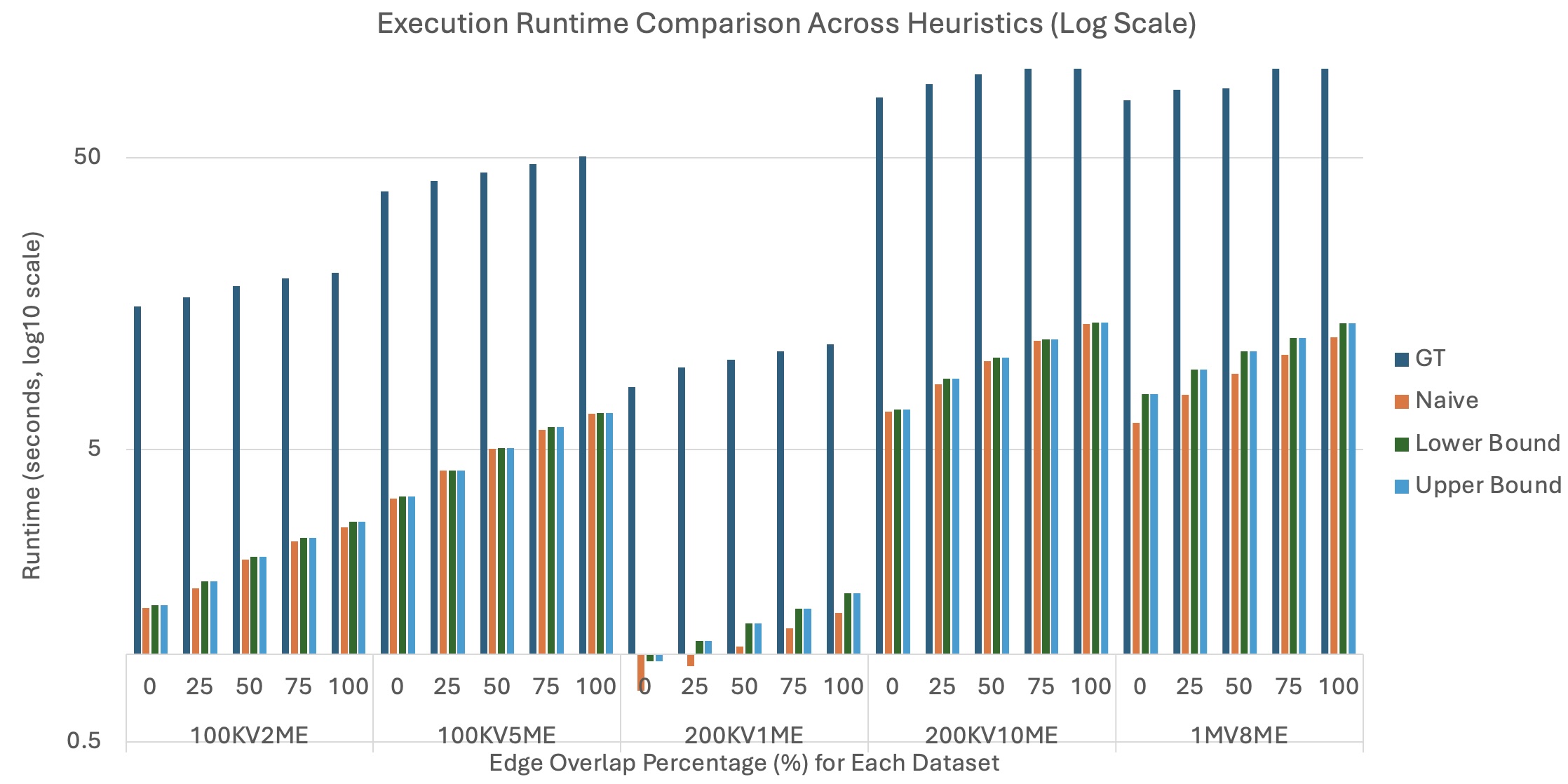}
    \caption{Comparison of Execution Time of the Strategies against Execution Time of Ground Truth for Synthetic Data Sets}
    \label{fig:max_total-time}
\end{figure}

\begin{figure}
    \includegraphics[width=\columnwidth]{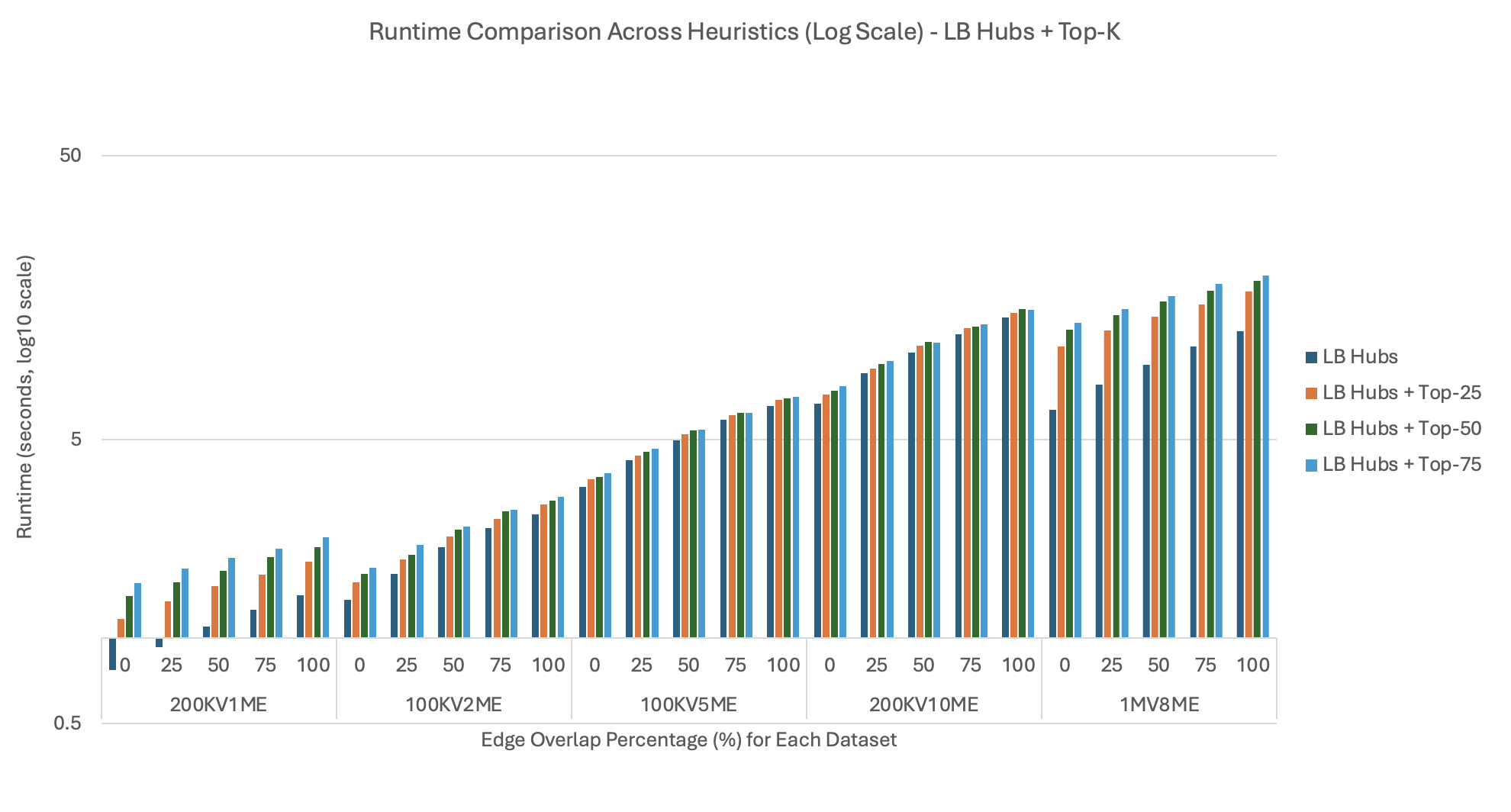}
    \caption{Runtime (log scale) for LB Hubs and Hubs + Top-$k$ across datasets and overlap levels}
    \label{fig:max_total-time-lb-hubs-topk}
\end{figure}

\begin{figure}
    \includegraphics[width=\columnwidth]{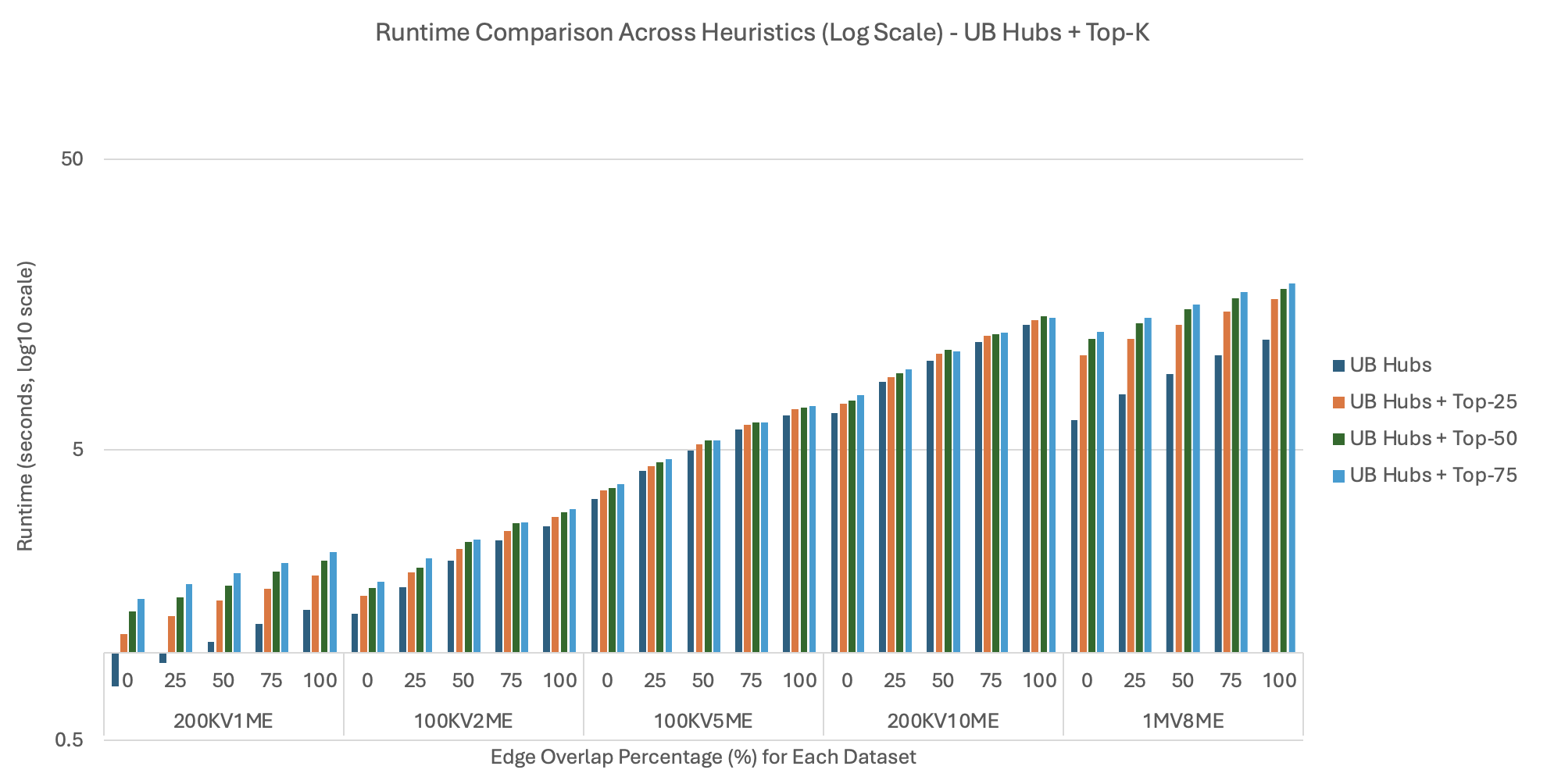}
    \caption{Runtime (log scale) for UB Hubs and Hubs + Top-$k$ across datasets and overlap levels}
    \label{fig:max_total-time-ub-hubs-topk}
\end{figure}
\underline{Runtime Performance:} 
Runtime results for the synthetic datasets are shown in Figures \ref{fig:max_total-time}, \ref{fig:max_total-time-lb-hubs-topk}, and \ref{fig:max_total-time-ub-hubs-topk}. All runtime figures are presented on a logarithmic scale to facilitate comparison across datasets with execution times that span multiple orders of magnitude. Figure \ref{fig:max_total-time} compares the ground-truth, Naïve, LB, and UB strategies, while Figures \ref{fig:max_total-time-lb-hubs-topk} and \ref{fig:max_total-time-ub-hubs-topk} illustrate the effects of hub-candidate selection on runtime.

Across all datasets and overlap configurations, ground-truth consistently incurs the highest runtime. This behavior is expected, as ground-truth requires explicit construction of the Boolean-OR aggregated graph and resolution of overlapping edges by selecting maximum weights, an edge-level process that becomes increasingly expensive for larger graphs and higher overlap levels. In contrast, the Naïve, LB, and UB heuristics operate directly on layer-wise strength values and avoid costly edge-level aggregation. As a result, all heuristic approaches achieve substantial runtime reductions relative to ground-truth. The LB and UB formulations exhibit nearly identical runtime behavior, reflecting the fact that both perform similar per node during composition.

Figure \ref{fig:max_total-time} also shows that overlap has a much larger effect on the ground-truth than on the heuristics. As overlap increases, ground-truth must process a greater number of shared edges when constructing the aggregated graph, resulting in increased execution time. The heuristic approaches remain comparatively insensitive to overlap, as they rely on precomputed layer-wise strengths rather than overlap-aware edge processing.

The effect of candidate-set selection is shown in Figures \ref{fig:max_total-time-lb-hubs-topk} and \ref{fig:max_total-time-ub-hubs-topk}. For both LB and UB, the Hubs-Only strategy consistently achieves the lowest runtime due to the small number of nodes considered during composition. Runtime increases as the candidate set expands from Top-$25$ to Top-$75$, reflecting both the additional cost of evaluating larger candidate sets during composition and the overhead associated with ranking and sorting nodes to identify the Top-$k$ candidates during analysis. Consequently, total runtime is not only influenced by the composition function but also by the analysis and node candidate selection stages that precede the composition. The increase in runtime is most pronounced when moving from Hubs-Only to Top-$25$, reflecting the additional analysis and sorting required to construct the candidate set. Subsequent expansion from Top-$25$ to Top-$75$ results in a more gradual increase in runtime, suggesting that the overhead associated with candidate selection remains modest relative to the edge-level aggregation costs incurred by the ground-truth approach.

Graph characteristics also influence runtime. Larger datasets, particularly 200KV10ME and 1MV8ME, exhibit the highest run-times across all strategies due to their substantially larger node and edge counts. While density contributes to computational cost, overall runtime is more strongly influenced by graph size. This is evident in the relatively sparse 1MV8ME dataset, which still incurs substantial runtime due to its large number of nodes and edges. 

\begin{figure}
    \includegraphics[width=\columnwidth]{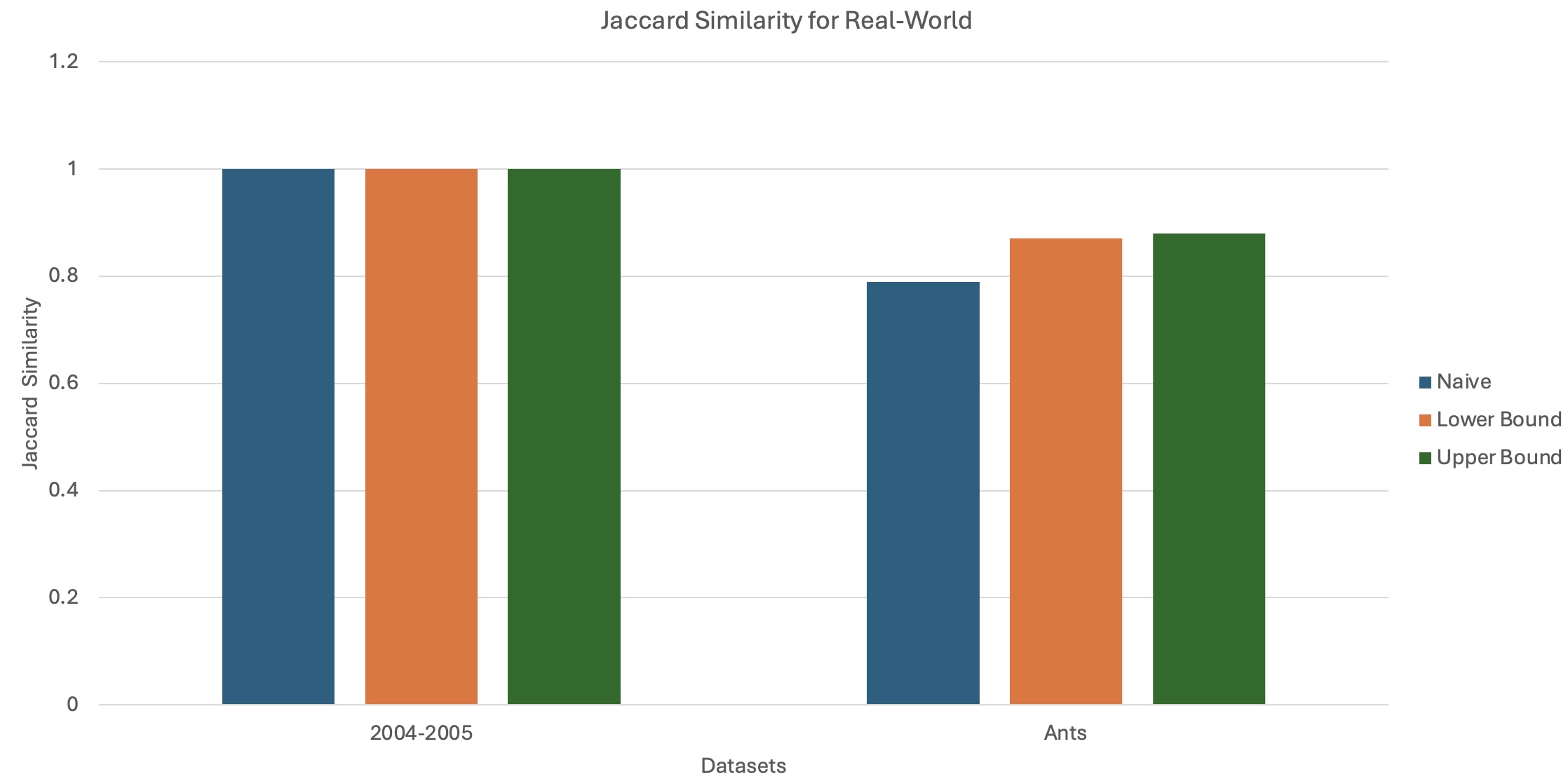}
    \caption{Accuracy Comparison Using Jaccard Similarity on Real-World Datasets}
    \label{fig:max_jaccard_rw}
\end{figure}

\underline{Real-World Datasets:}
Figure \ref{fig:max_jaccard_rw} presents the accuracy of the heuristics on the real-world datasets. Overall, the results are consistent with the trends observed on the synthetic datasets, with the UB heuristic achieving the highest Jaccard similarity, followed by LB and Naïve. 

For the co-authorship dataset (2004 - 2005), all strategies achieve near-perfect Jaccard similarity. This indicates that the dominant hubs are strongly preserved across layers and remain highly distinguishable after aggregation. Under such conditions, even approximate strength estimates are sufficient to recover the ground-truth hub set, resulting in minimal differences between the heuristics. In contrast, the Ants datasets exhibits greater sensitivity to the choice of heuristic. The Naïve strategy achieves the lowest accuracy, while both bound-based heuristics provide improved hub identification. The superior performance of UB suggests that preserving strength contributions from both layers is beneficial for identifying dominant nodes under max aggregation. Compared with the co-authorship network, the Ants dataset appears to be more sensitive to small changes in node strength, resulting in larger differences between the heuristics.   

Overall, the real-world results reinforce the observations from the synthetic datasets. Heuristics that preserve layer-wise strength information consistently outperform approaches based solely on hub membership, with the UB formulation providing the closest approximation to the ground-truth hub set. 

\underline{Runtime Performance:} 
Runtime results for the real-world datasets are summarized in Table \ref{table:max_rw_runtime}. Similar to the synthetic experiments, ground-truth incurs the highest runtime due to the cost of explicit overlap-aware graph aggregation. In contrast, the heuristic approaches avoid edge-level aggregation and operate directly on layer-wise strength values, resulting in substantially lower execution times. 

On the co-authorship dataset, both LB and UB reduce total runtime by more than a factor of three relative to ground truth while maintaining near-perfect accuracy. The Naïve strategy achieves the lowest runtime overall because it performs the least amount of composition work. On the Ants dataset, run-times are uniformly low due to the small graph size, however, ground-truth remains the most expensive approach, while the heuristic strategies complete nearly instantaneously. These results demonstrate that the computational advantages observed on the synthetic datasets also hold for the real-world networks. By eliminating the need for explicit overlap-aware aggregation, the proposed heuristics provide substantial runtime reductions while maintaining high accuracy in recovering the ground-truth hub set. 

\begin{table}[!htp]
    \caption{Total runtime (in seconds) for real-world datasets}
    \label{table:max_rw_runtime}
        \begin{tabular}{| c | c | c | c | c |} \hline
            \textbf{Dataset} & \textbf{GT} & \textbf{Lower Bound} & \textbf{Upper Bound} & \textbf{Naïve}   \\ \hline
            2004-2005        & 1.575       & 0.439                & 0.439       & 0.204             \\ \hline
            Ants             & 0.057       & 0.036                & 0.036       & 0.034           \\ \hline
        \end{tabular}

\end{table}
\section{Weighted Degree Centrality Using Sum and Boolean-OR Aggregation}
\label{section:sum}
In contrast to max aggregation, sum aggregation yields an additive form for the ground truth strength, which makes decoupled composition exact using only layer-wise strengths. As before, we use the same two-function decoupling approach: (i) an analysis function that processes each layer independently, and (ii) a composition function that combines the layer-wise results using heuristics. Under sum aggregation, this modularization not only improves efficiency and scalability, but also enables controlled composition using partial layer-wise information. 

\subsection{Definitions}
Let the two layers of an HoMLN be denoted as $Layer_x$ and $Layer_y$, with a shared node set $V$. For any node $u \in V$, the analysis function computes its layer-wise strength (weighted degree) by summing the weights on edges incident on a node. 

\begin{equation*}
    s_x(u) = \sum_{(u, v) \in E_x} w_x(u, v), \quad
    s_y(u) = \sum_{(u, v) \in E_y} w_y(u, v)
\end{equation*}

\underline{Ground Truth Using Boolean-OR and Sum:} The ground truth graph $G_{GT}$ is formed using \textit{Boolean-OR} on the edge structure (if the edge is present in at least one layer) and \textit{sum} aggregation for the weights on those edges. Using this, the ground truth strength of node $u$ is: 
\begin{equation*}
    s_{GT}(u) = s_x(u)+s_y(u)
\end{equation*}

The corresponding global average strength: 
\begin{equation*}
    \bar{s}_{GT} = \frac{1}{|V|} \sum_{u \in V} s_{GT}(u) 
         = \frac{1}{|V|} \sum_{u \in V}(s_x(u)+s_y(u))
\end{equation*}

A node $u$ is considered a ground truth hub if $s_{GT}(u) > \bar{s}_{GT}$.

\underline{Strength Estimate Under the Decoupling Approach:} Since the ground truth strength is the sum of the layer-wise strengths, the aggregated strength of a node can be directly computed during composition by adding its layer strengths, without having to explicitly construct the aggregated graph. 

\begin{equation*}
    \hat{s}(u) = s_x(u) + s_y(u)
\end{equation*}

Additionally, if the analysis function carries forward the layer strength sums $\text{sum}_x = \sum_{u \in V}s_x(u)$, \quad $\text{sum}_y = \sum_{u \in V}s_y(u)$, then the global average can be computed as:

\begin{equation*}
    \bar{s} =\frac{\text{sum}_{x} + \text{sum}_y}{|V|}
\end{equation*}

which allows for the correct threshold calculation during composition. 

\subsection{Lemmas and Intuition}
Since sum aggregation yields a simple form for $s_{GT}(u)$, it also produces characteristics when heuristic composition is restricted to a small candidate set. 

\begin{lemma}
\label{lemma:sum-lemma-a}
A node that is a hub in only one layer of a weighted homogeneous multilayer network (HoMLN) may not remain a hub in the ground truth network constructed using Boolean-OR aggregation with sum aggregation.

\end{lemma}
\begin{proof}
Let the two layers be $Layer_x$ and $Layer_y$ with common node set $V$, and let $s_{GT}(u) = s_x(u) + s_y(u)$ denote  the strength of node $u$ in the Boolean-OR  ground truth with sum aggregation. Let $\bar{s}_{GT}$, $\bar{s}_{x}$, $\bar{s}_{y}$ denote the corresponding average strengths.

Assume a node $u$ is a hub in exactly one layer, for example $Layer_x$, so that $s_x(u) > \bar{s}_x$ and $s_y(u) \le \bar{s}_y$.

Node $u$ will fail to be a hub in the ground truth whenever

\begin{equation}
\label{eqn:lemma-a_cond1}
    s_x(u) + s_y(u) \le \bar{s}_{GT}
\end{equation}

Since $\bar{s}_{GT} = \bar{s}_x + \bar{s}_y$, after substituting and rearranging, equation \ref{eqn:lemma-a_cond1} is equivalent to 
\begin{equation}
\label{eqn:lemma-a_cond2}
    s_x(u) - \bar{s}_x \le \bar{s}_y - s_y(u)
\end{equation}
So, even if $u$ exceeds the average in one layer by a small margin, it can still fall below the ground truth threshold if its strength in the other layer is sufficiently below the average. This situation can occur in weighted HoMLNs because layer-wise strength distributions may differ and a node can have weak or zero contributions in one layer. Consequently, a node being a hub in only one layer does not guarantee that it remains a hub after sum aggregation.  
\end{proof}

\begin{lemma}
\label{lemma:sum-lemma-b}
A node that is not a hub in one of the layers of a weighted homogeneous multilayer network (HoMLN) may become a hub in the ground truth network constructed using Boolean-OR aggregation with sum aggregation for the weights.
\end{lemma}
\begin{proof}
Let the two layers be $Layer_x$ and $Layer_y$ with common node set $V$, and let $s_{GT}(u) = s_x(u) + s_y(u)$ denote  the strength of node $u$ in the Boolean-OR  ground truth with sum aggregation. Let $\bar{s}_{GT}$, $\bar{s}_{x}$, $\bar{s}_{y}$ denote the corresponding average strengths. Since averages are linear under the sum aggregation, 

\begin{equation*}
\bar{s}_{GT} =  \frac{1}{|V|} \sum_{u \in V}(s_x(u)+s_y(u)) = \bar{s}_x + \bar{s}_y
\end{equation*}

Assume a node $u$ is not a hub in one layer, for example $Layer_y$, so that $s_y(u) \le \bar{s}_y$. Node $u$ becomes a hub in the ground truth whenever

\begin{equation*}
    s_{GT}(u) = s_x(u)+s_y(u) > \bar{s}_{GT} = \bar{s}_x + \bar{s}_y
\end{equation*}

which is equivalent to 
\begin{equation*}
    s_x(u) - \bar{s}_x > \bar{s}_y - s_y(u)
\end{equation*}
So, even if the node $u$ does not exceed the hub threshold in $Layer_y$, it can still be a ground truth hub, given that its strength in $Layer_x$ is large enough to compensate for its deficit in $Layer_y$. Therefore, a node does not need to be a hub in every layer to be identified as a hub during composition. 
\end{proof}

\begin{lemma}
If a node exceeds the layer-average strength in each layer, then it remains a hub in the Boolean-OR with sum aggregation ground truth. 
\label{lemma:sum-lemma-c}
\end{lemma}

\begin{proof}
    Let $Layer_x$ and $Layer_y$ be two layers with common node set $V$. 
    For $u \in V$, let $s_x(u)$ and $s_y(u)$ denote its strength in layers $Layer_x$ and $Layer_y$, respectively, and define 
    \begin{equation*}
        \bar{s}_x = \frac{1}{|V|}\sum_{u\in V}s_x(u), \qquad 
        \bar{s}_y = \frac{1}{|V|}\sum_{u\in V}s_y(u)
    \end{equation*}
    If $s_x(u) > \bar{s}_x$ and $s_y(u) > \bar{s}_y$, then 
    \begin{equation*}
        s_{GT}(u) = s_x(u) + s_y(u) > \bar{s}_x + \bar{s}_y
    \end{equation*}
    
    In the composition function, the combined strength is $s(u) = s_x(u) + s_y(u)$ and the global average is
    \begin{equation*}
       \bar{s} = \frac{1}{|V|}\sum_{u\in V}s(u) = \frac{1}{|V|}\sum_{u \in V}(s_x(u) + s_y(u)) = \bar{s}_x + \bar{s}_y 
    \end{equation*}
    Hence,
    \begin{equation*}
        s(u) = s_x(u) + s_y(u) > \bar{s}_x + \bar{s}_y = \bar{s}
    \end{equation*}
    which implies
    \begin{equation*}
        s(u) > \bar{s},
    \end{equation*}
    Therefore, $u$ remains a hub in the composition. 
\end{proof}

\subsection{Composition Strategies}
We evaluate three strategies that utilizes different sizes and selection of node from analysis function to be used in the composition function. 

\underline{(1) Hubs-Only.} The candidate set contains only nodes identified as hubs in $Layer_x$ or $Layer_y$. During composition, $\hat{s}$ is evaluated only for this restricted set and nodes exceeding the global threshold $\bar{s}$ are retained. The threshold $\bar{s}$ is computed using the full layer-wise strength sums, i.e., \begin{equation*}
    \bar{s} = \dfrac{(\text{sum}_x + \text{sum}_y)}{|V|}
\end{equation*} independent of the candidate set size. By evaluating only the candidate hub set instead of the entire node set, this strategy reduces the number of nodes evaluated during composition while applying the same strength computation and global hub threshold used in the ground-truth formulation.

The Hubs-Only composition strategy is formalized in Algorithm \ref{alg:sum-hubs_only}, where only the union of the layer-wise hub sets is used as the candidate set for composition. 
\begin{algorithm}[tbp]
\caption{Hubs-Only Composition Strategy}\label{alg:sum-hubs_only}
\begin{algorithmic}[1]
\REQUIRE Hub sets $H_x$, $H_y$; strengths $s_x(u)$, $s_y(u)$; sums $\text{sum}_x$, $\text{sum}_y$; total node count $\lvert V\rvert $
\STATE $H \leftarrow \emptyset$
\STATE $\bar{s} \leftarrow \dfrac{\text{sum}_x +\text{sum}_y}{\lvert V\rvert }$
\STATE $C \leftarrow H_x \cup H_y$
\FOR{each $u \in C$}
\STATE $\hat{s}(u) \leftarrow s_x(u) + s_y(u)$
\IF{$\hat{s}(u) > \bar{s}$}
\STATE $H \leftarrow H \cup \{u\}$
\ENDIF
\ENDFOR
\RETURN $H$
\end{algorithmic}
\end{algorithm}

\underline{(2) All Nodes (Exact Node Strength).} The composition step evaluates every node in $V$ using $\hat{s} = s_x(u) + s_y(u)$. Since this exactly matches $s_{GT}(u)$ under sum aggregation, this strategy guarantees perfect alignment with the ground truth set. The global threshold computed using 
\begin{equation*}
    \bar{s} = \frac{(\text{sum}_x + \text{sum}_y)}{|V|}
\end{equation*}
Unlike explicit aggregation, this strategy does not require constructing edge set, it relies only on layer-wise strengths and their sums. However, evaluating all nodes during composition incurs high computational cost. 

The complete procedure for decoupled hub composition using the All Nodes strategy for sum aggregation is formalized in Algorithm \ref{alg:sum-all_nodes}.

\begin{algorithm}[tbp]
\caption{All Nodes Composition Strategy}\label{alg:sum-all_nodes}
\begin{algorithmic}[1]
\REQUIRE Layer-wise strengths $s_x(u)$ and $s_y(u)$ for all $u \in V$ and sums $\text{sum}_x$, $\text{sum}_y$
\STATE $H \leftarrow \emptyset$
\STATE $\bar{s} \leftarrow \dfrac{\text{sum}_x +\text{sum}_y}{|V|}$
\FOR{each $u \in V$}
\STATE $\hat{s}(u) \leftarrow s_x(u) +s_y(u)$
\IF{$\hat{s}(u) > \bar{s}$}
\STATE $H \leftarrow H \cup \{u\}$
\ENDIF
\ENDFOR
\RETURN $H$
\end{algorithmic}
\end{algorithm}

\underline{(3) Top-$k$ Nodes.} In this strategy, the candidate set consists of the top-$k$ highest-strength nodes selected independently from each layer, regardless of whether they are identified as local-layer hubs. Composition is performed only on this subset, while the global threshold $\bar{s}$ is still computed using the full layer-wise strength sums. This strategy captures nodes that may not exceed the hub threshold in any individual layer but become hubs after strengths are aggregated as described in Lemma \ref{lemma:sum-lemma-b}. The parameter $k$ provides a tunable trade-off between accuracy and efficiency: larger $k$ values improve coverage of potential hub identification while avoiding a full scan of all nodes. 

\subsection{Experimentation Results}
The proposed strategies were evaluated using runtime, precision, and accuracy. Experiments are conducted on the synthetic datasets with different configurations, like edge overlap percentages, as well as the real-world datasets. Accuracy is measured using Jaccard similarity against the ground truth hub set, while runtime captures the cost of analysis, composition, and ground truth aggregation. 

\begin{figure}[!ht]
    \centering
    \includegraphics[width=\columnwidth]{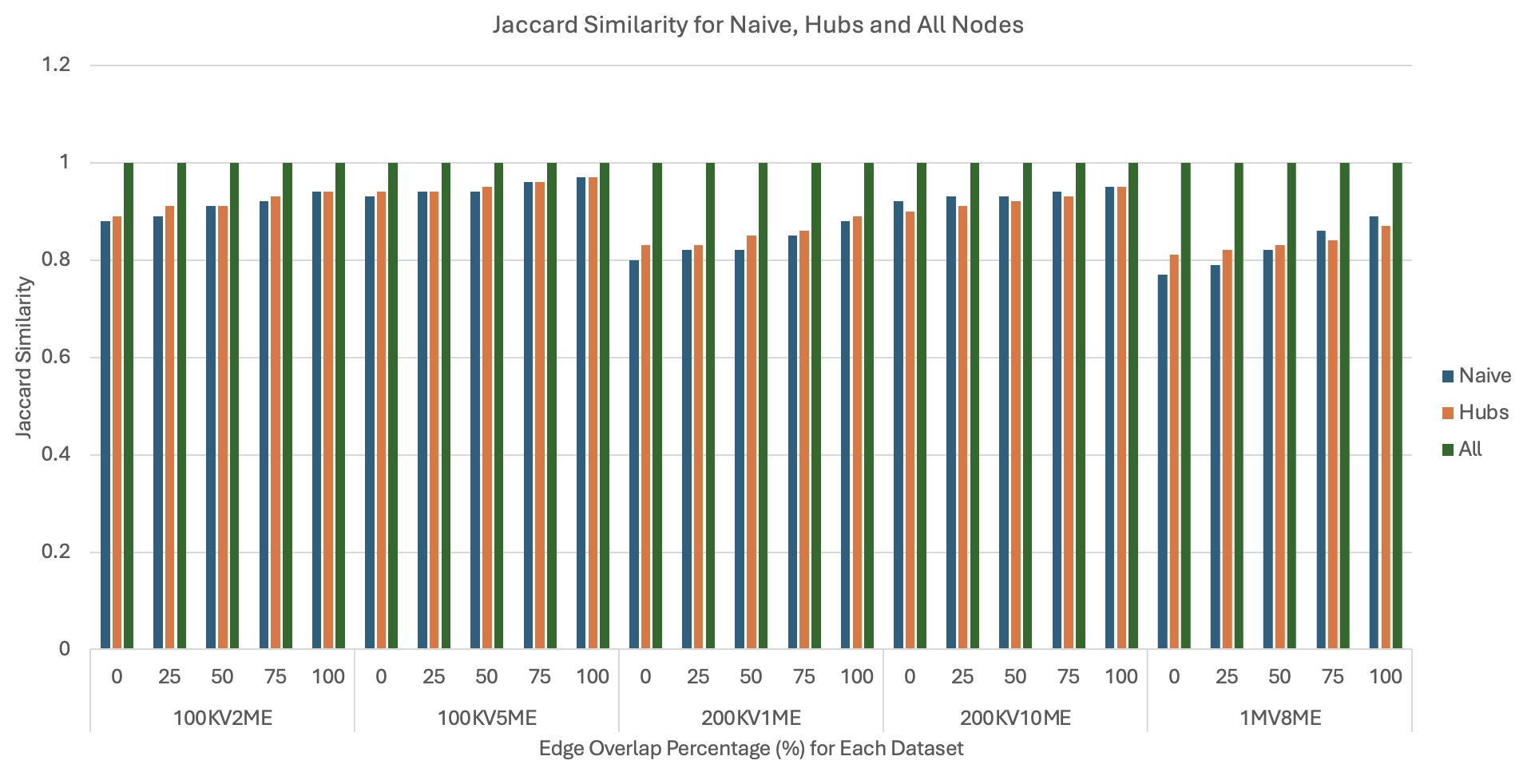}
    \caption{Accuracy Comparison Using Jaccard Similarity for Naive, Hub-Only, All-Nodes Strategies}
    \label{fig:sum-jaccard-naïve-hubs-all}
\end{figure}

\begin{figure}[!ht]
    \includegraphics[width=\columnwidth]{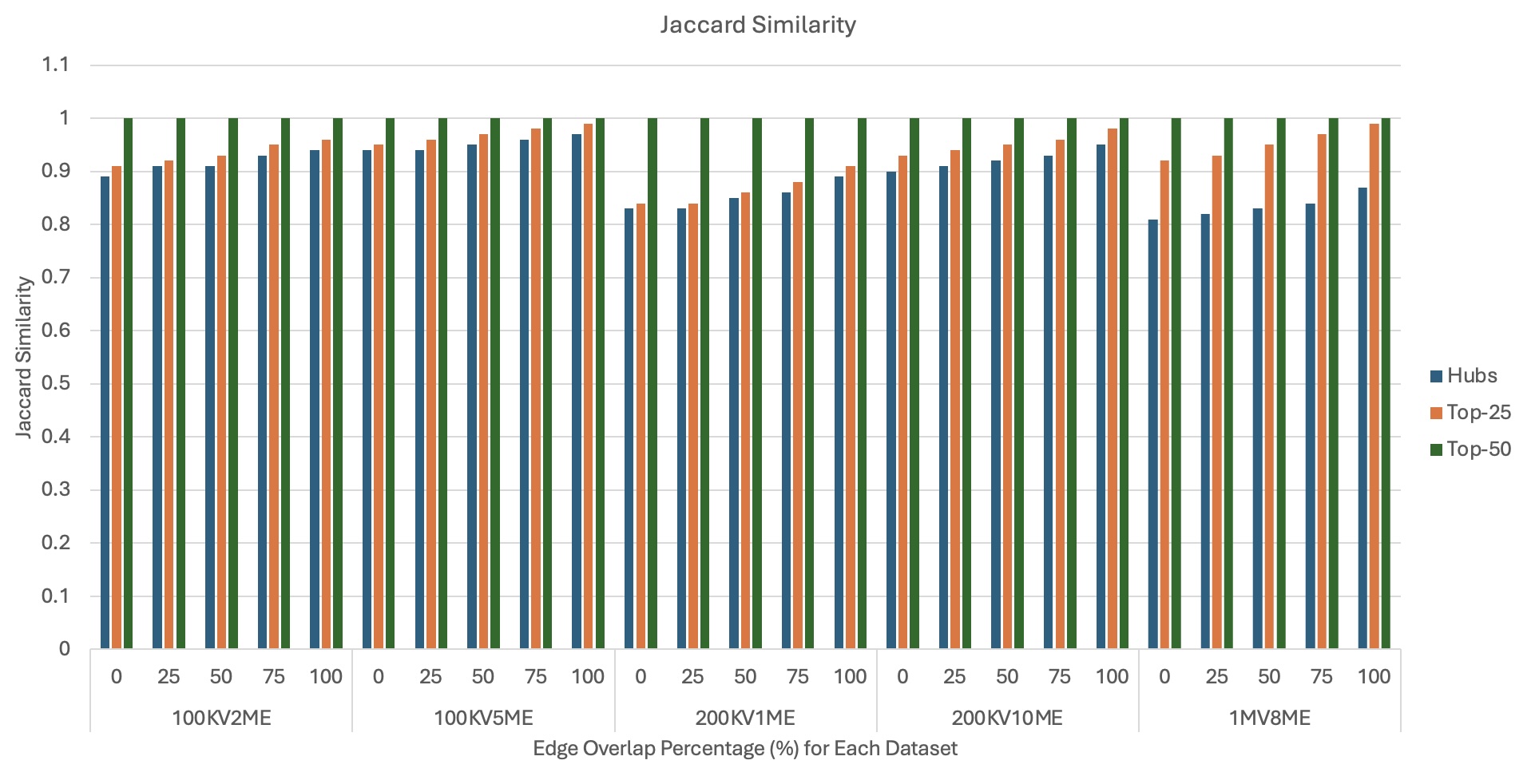}
    \caption{Accuracy Comparison Using Jaccard Similarity for Top-$k$ strategies}
    \label{fig:sum-jaccard-topk}
\end{figure}

\underline{Accuracy and Hub Identification:} Accuracy results are shown in Figures \ref{fig:sum-jaccard-naïve-hubs-all} and \ref{fig:sum-jaccard-topk}. Keeping and using all the nodes achieves perfect Jaccard similarity score with the ground truth hub set across all datasets and overlap configurations. This is expected, since under the sum aggregation $\hat{s}(u) = s_{GT}(u)$ and the global threshold is computed exactly. 

The Naïve and Hubs only strategies demonstrates a reduced accuracy, especially at low overlap levels. This behavior is consistent with Lemmas \ref{lemma:sum-lemma-a} and \ref{lemma:sum-lemma-b}, which show that hub membership under the sum aggregation cannot be reliably deduced from local layer hub identification alone. 

The Top-$k$ strategy substantially improves the accuracy of the hub identification. Selecting the top 25\% of nodes from each layer significantly increases Jaccard similarity compared to the Hubs-only strategy. Notably, selecting the top 50\% of nodes consistently achieves a prefect Jaccard similarity across all datasets and overlaps. This observation indicates that true ground truth hubs are reliably contained within the upper half of the local layer strength rankings.

Edge overlap between layers also plays an important role in hub identification accuracy. As overlap increases, a larger fraction of edges contribute weights to the same node pair across layers, which reinforces the total node strength in the ground truth graph. This reduces the discrepancies between layer rankings, making it easier to identify hubs during composition. Across all strategies, accuracy improves as overlap increases, and the gap between using all the nodes and the other strategies narrows.

\begin{figure}[!ht]
    \includegraphics[width=\columnwidth]{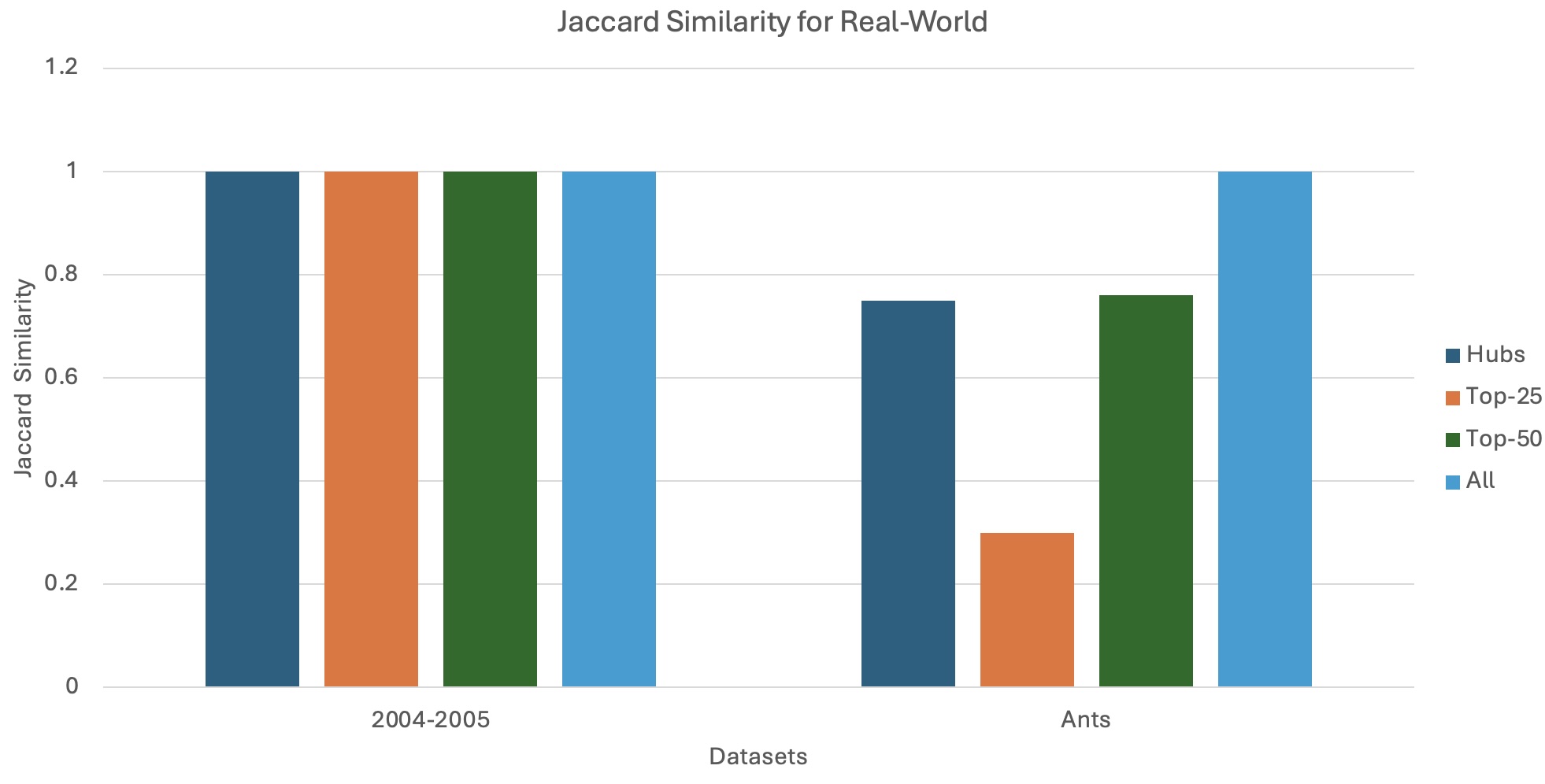}
    \caption{Accuracy Comparison Using Jaccard Similarity on Real-World Datasets}
    \label{fig:sum_jaccard_rw}
\end{figure}

\underline{Real-World Accuracy Evaluation:} To validate the proposed composition strategies on the real-world datasets, we evaluate their performance on the two real-world datasets, the 2004-2005 Co-Authorship and the Ants networks, as shown in figure \ref{fig:sum_jaccard_rw}.

For the co-authorship dataset, all strategies achieve perfect Jaccard similarity, indicating that the ground truth hubs are consistently recovered regardless of the candidate selection strategy. This behavior reflects the strong concentration of interaction weights and stable hub structure across layers in the co-authorship network. 

In contrast, the Ants dataset exhibits greater sensitivity to candidate selection. While the All-Nodes strategy achieves perfect recovery, the Hubs-only and Top-$k$ strategies display varying accuracy, with Top-$50$ substantially outperforming Top-25. This highlights the impact of uneven distribution across nodes and weaker hub dominance in the Ants network. Nevertheless, selecting a sufficiently large top-ranked subset reliably recovers the majority of ground truth hubs, reinforcing the effectiveness of partial ranking strategies in real-world scenarios. 

\begin{figure}[ht]
    \includegraphics[width=\columnwidth]{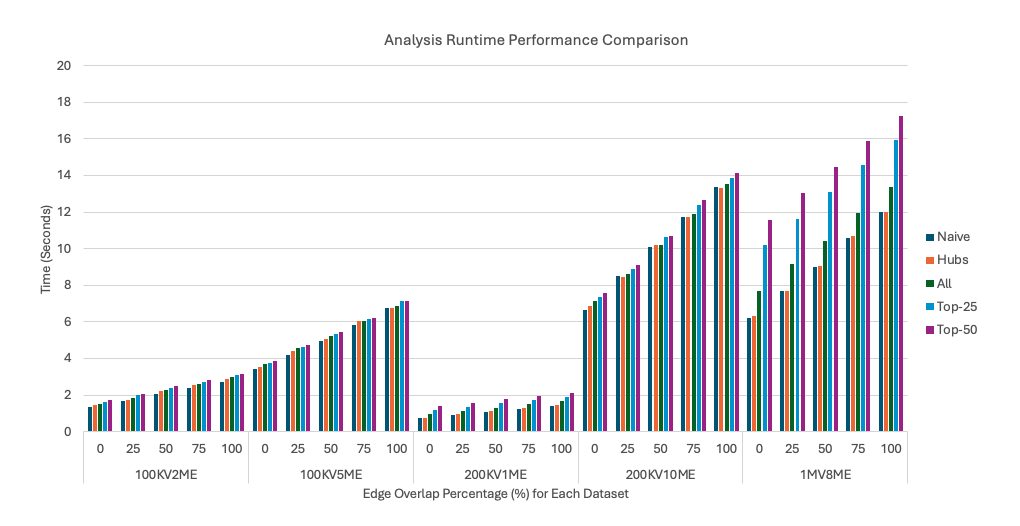}
    \caption{Layer-Wise Analysis Runtime for Weighted for Synthetic Data sets}
    \label{fig:sum-analysis-time}
\end{figure}

\begin{figure}[ht]
    \includegraphics[width=\columnwidth]{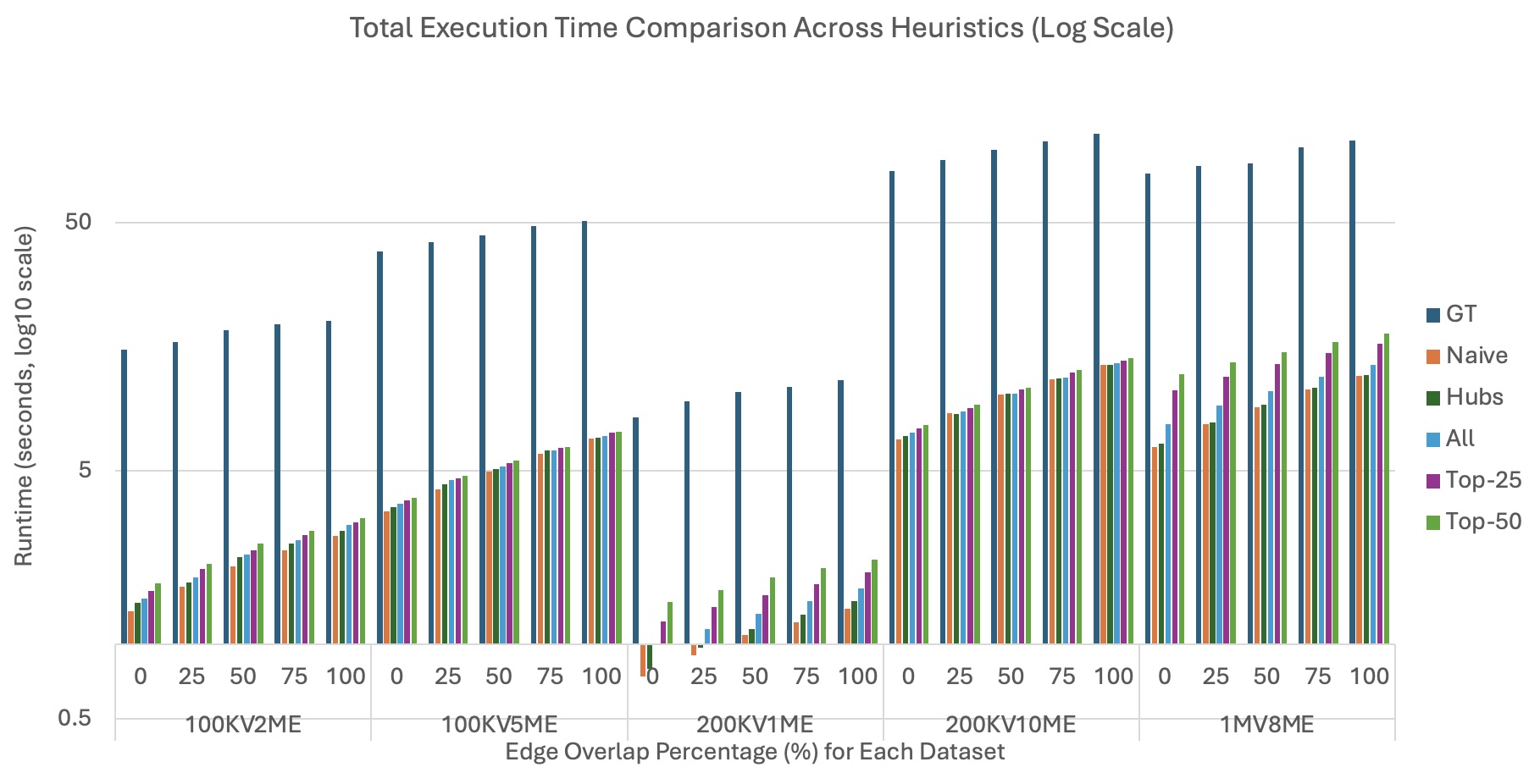}
    \caption{Comparison of Execution Time of the Strategies against Execution Time of Ground Truth
for Synthetic Data Sets}
    \label{fig:sum-total-time}
\end{figure}

\underline{Runtime Performance:} Runtime results for the proposed composition strategies are summarized in Figures \ref{fig:sum-analysis-time} and \ref{fig:sum-total-time}, with detailed composition times for the 50-50 distribution split in Table \ref{table:sum_composition_result_50-50}. Since the core layer-wise strength computation is common to all strategies, observed difference arise from additional processing during the analysis and composition functions.

The All-Nodes strategy consistently exhibits the lowest composition time across all datasets and overlap configurations. This behavior is expected, as the strategy evaluates every node using a singe linear scan over array-based strength representations, avoiding candidate filtering or sorting. Despite processing all nodes, this implementation minimizes overhead and yields highly efficient composition. In contrast, the Naïve and Hubs-Only strategies incur higher composition times despite evaluating fewer nodes. The Naïve relies on set union and Hubs-only on dictionary lookups to construct and process the candidate sets during composition, which introduces additional overhead.

The Top-$k$ strategy exhibits the highest runtime among the heuristic approaches due to additional processing during the analysis function. In particular, identifying the top ranked nodes in each layer requires maintaining a heap-based priority queue over node strengths. This ranking overhead increases the analysis time and grows with both dataset size and the value of $k$. As a result, Top-$k$ strategies incur higher analysis costs than Hubs-Only or Naïve strategies, even before composition. 

Across all strategies, analysis time increases as edge overlap grows. Higher overlap produces larger accumulated node strengths, increasing the cost of strength computation and ranking operations. This effect is most pronounced for the Top-$k$ strategies, where sorting or heap maintenance must process larger or more closely ranked strength values as overlap increases.

Dataset size has a clear impact on runtime, with larger graphs incurring higher costs across all strategies. To further validate runtime behavior in practical settings, we evaluate the proposed strategies on real-world datasets. Table \ref{table:sum_rw_runtime} summarizes the total runtime for each strategy on the co-authorship and Ants dataset. Across both datasets, all heuristic strategies achieve substantial runtime reductions compared to explicit ground truth aggregation. The Ants dataset incurs extremely low run-times across all strategies due to its smaller size, while the co-authorship network exhibits higher but still manageable costs. Notably, Top-$k$ strategies incur additional overhead relative to Hubs-Only due to ranking operations, yet remain significantly faster than the explicit ground truth computation.





\begin{table}[!ht]
    \caption{Total runtime (in seconds) for real-world datasets}
    \label{table:sum_rw_runtime}
    \resizebox{\columnwidth}{!}{%
        \begin{tabular}{| c | c | c | c | c | c | c |}
			\hline
\textbf{Dataset} & \textbf{GT} & \textbf{Hubs} & \textbf{All} & \textbf{Top-25} & \textbf{Top-50}  \\ \hline
2004-2005        & 1.547 & 0.208   & 0.438  & 0.591     & 0.734      \\ \hline
Ants             & 0.057 & 0.036  & 0.0486  & 0.036     & 0.036   \\ \hline  
\end{tabular}
}

\end{table}



Overall, these results highlight a clear tradeoff between runtime and accuracy. While the All-Node strategy provides exact composition with minimal overhead, the Top-$k$ strategy incurs additional analysis cost due to ranking but achieves near-perfect or perfect hub recovery while evaluating only a fraction of the nodes. This demonstrates that accurate hub identification can be achieved using only partial layer-wise ranking information, with the Top-$k$ strategy recovering the ground truth hubs while evaluating only a subset of nodes.

\begin{table}[!ht]
    \caption{Composition runtime (in Seconds) for Naïve, All, Hubs, and Top-$k$ strategies on the 50-50 distribution datasets}
	\label{table:sum_composition_result_50-50}
    \resizebox{\columnwidth}{!}{%
        \begin{tabular}{| c | c | c | c | c | c | c |}
			\hline
                        \textbf{Dataset} & \textbf{Overlap}\% & \textbf{Naïve} & \textbf{All} & \textbf{Hubs}    & \textbf{Top-25}  & \textbf{Top-50}  \\ \hline
\textbf{100KV2ME}  & 0          & 0.00346        & 0.00041      & 0.00886 & 0.01411 & 0.03044 \\ \hline
                   & 25         & 0.00344        & 0.00041      & 0.00947 & 0.01404 & 0.03049 \\ \hline
                   & 50         & 0.00342        & 0.00037      & 0.00971 & 0.01428 & 0.02986 \\ \hline
                   & 75         & 0.00339        & 0.00036      & 0.00933 & 0.01445 & 0.02831 \\ \hline
                   & 100        & 0.00339        & 0.00041      & 0.00937 & 0.01421 & 0.02952 \\ \hline
\textbf{100KV5ME}  & 0          & 0.00338        & 0.00045      & 0.00921 & 0.01431 & 0.02844 \\ \hline
                   & 25         & 0.00339        & 0.00041      & 0.00926 & 0.01384 & 0.03071 \\ \hline
                   & 50         & 0.00342        & 0.00042      & 0.00934 & 0.01368 & 0.03106 \\ \hline
                   & 75         & 0.00341        & 0.00044      & 0.00912 & 0.01391 & 0.02944 \\ \hline
                   & 100        & 0.00341        & 0.00045      & 0.00927 & 0.01394 & 0.03043 \\ \hline
\textbf{200KV1ME}  & 0          & 0.00560        & 0.00091      & 0.02098 & 0.03363 & 0.06014 \\ \hline
                   & 25         & 0.00560        & 0.00093      & 0.02053 & 0.03199 & 0.06225 \\ \hline
                   & 50         & 0.00566        & 0.00090      & 0.02059 & 0.03314 & 0.06194 \\ \hline
                   & 75         & 0.00534        & 0.00091      & 0.01905 & 0.03181 & 0.06220 \\ \hline
                   & 100        & 0.00555        & 0.00090      & 0.01958 & 0.03173 & 0.06323 \\ \hline
\textbf{200KV10ME} & 0          & 0.00547        & 0.00092      & 0.01948 & 0.03341 & 0.06748 \\ \hline
                   & 25         & 0.00548        & 0.00092      & 0.01826 & 0.03386 & 0.06590 \\ \hline
                   & 50         & 0.00547        & 0.00091      & 0.01954 & 0.03335 & 0.06595 \\ \hline
                   & 75         & 0.00548        & 0.00089      & 0.01927 & 0.03103 & 0.06459 \\ \hline
                   & 100        & 0.00545        & 0.00088      & 0.01913 & 0.03185 & 0.06723 \\ \hline
\textbf{1MV8ME}    & 0          & 0.01754        & 0.00864      & 0.14960 & 0.33157 & 0.62571 \\ \hline
                   & 25         & 0.01733        & 0.00999      & 0.14426 & 0.33294 & 0.61673 \\ \hline
                   & 50         & 0.01711        & 0.01030      & 0.14361 & 0.33882 & 0.61575 \\ \hline
                   & 75         & 0.01694        & 0.00718      & 0.13530 & 0.33459 & 0.62284 \\ \hline
                   & 100        & 0.01659        & 0.00652      & 0.13075 & 0.33471 & 0.62147 \\ 
			\hline
		\end{tabular}
		}
\end{table}

\section{Conclusion}
\label{section:conclusion}

In this paper, we studied weighted degree centrality in homogeneous multilayer networks (HoMLNs) under Boolean-OR aggregation of edge structure and two commonly used weight aggregation operators: \textit{max} and \textit{sum}. Using a decoupling-based framework, we separated computation into a layer-wise analysis function and a heuristic-based composition function, enabling hub identification without explicitly constructing the aggregated multilayer graph. 

We formally characterized the behavior of weighted degree centrality under both aggregation operators through definitions, bounds, and lemmas that highlight the fundamental differences between additive and non-additive aggregation. For sum aggregation, we showed that the ground truth node strength can be computed exactly from layer-wise strengths, without explicitly constructing the aggregated graph. As a result, hub identification during the composition function is exact when all nodes are evaluated. 

For max aggregation, we showed that exact ground truth strength cannot, in general, be computed from layer-wise strength alone due to its dependence on edge-level overlap across layers. To address this limitation, we proposed a conservative lower bound heuristic and an optimistic upper bound heuristic that provide bounds on the true node strength. Our analysis explain why dominance in at least one layer plays a critical role in hub identification. 

Extensive experimental evaluation on synthetic datasets with varying edge distributions and edge overlaps configurations, as well as, real-world networks, validates the theoretical analysis. The results show that the proposed heuristics achieve high hub identification accuracy while providing significant improvements in efficiency compared to ground truth aggregation. In particular, ranking based strategies consistently recovers the majority of ground truth hubs while evaluating only a subset of nodes.  

Overall, this work demonstrates that the decoupling-based framework, when combined with heuristics, offers a practical and scalable approach for centrality analysis in weighted HoMLNs. The results confirm that accurate hub identification can be achieved using partial layer-wise information, significantly reducing computational overhead. 

Future work includes extending the proposed heuristics to networks with more than two layers, where aggregation effects compound across layers and trade-offs between accuracy and runtime are more pronounced. Additional directions include exploring other aggregation methods like \textit{min} and \textit{average}, as well as extending the framework to other centrality measures such as weighted closeness or betweenness centrality.

\bibliographystyle{IEEEtran}
\bibliography{bibliography/itlabTheses.bib, bibliography/itlabCollectionPart1, bibliography/itlabCollectionPart2, bibliography/itlabPublications, bibliography/newbibs}

\end{document}